\documentclass[12pt,preprint]{elsarticle}
\usepackage[utf8x]{inputenc}

\usepackage{amsmath,amssymb}
\RequirePackage{pgf,tikz}
\usepackage{multirow}
\usepackage[]{algorithm2e}
\usepackage{rotating}
\usepackage{enumerate}
\usepackage{array}
\usepackage[labelformat=simple]{subfig}
\usepackage{float}

\usepackage{color}
\usepackage{colortbl}
\usepackage{wasysym}
\usepackage{mathrsfs}
\usepackage{enumerate}
\usepackage{pxfonts}
 \usepackage{amsfonts}
\usepackage{dsfont}
\usepackage{array}
\usepackage{lscape}
\usepackage[pdftex,                %
    bookmarks         = true,
    bookmarksnumbered = true,
    pdfpagemode       = None,
    pdfstartview      = FitH,
    pdfpagelayout     = SinglePage,
    colorlinks        = true,
    urlcolor          = magenta,
    pdfborder         = {0 0 0}
    ]{hyperref}%

\usepackage{tabularx,ragged2e,booktabs,caption}
\newcolumntype{Y}{>{\centering\arraybackslash}X}
\newtheorem{The}{Theorem}[section]
\newtheorem{Pro}[The]{Proposition}

\newtheorem{Ex}[The]{Example}
\newtheorem{Def}[The]{Definition}
\newtheorem{Cor}[The]{Corollary}

\newtheorem{proof}{Proof}

\usepackage{color}
\definecolor{darkred}{rgb}{.7, 0, 0}

\numberwithin{equation}{section}
\numberwithin{table}{section}
\numberwithin{figure}{section}
\usepackage{geometry}

\journal{\textbf{arXiv}}

\begin{document}
\begin{frontmatter}
\title{Bayesian Bandwidths in Semiparametric Modelling for Nonnegative Orthant Data with Diagnostics }
\author[rvt]{C\'elestin C. Kokonendji}
\ead{celestin.kokonendji@univ-fcomte.fr}
\author[rvt1,rvt2]{Sobom M. Som\'e\corref{cor1}}
\ead{sobom.some@gmail.com}
\cortext[cor1]{\textit{Corresponding author:} Universit\'e Thomas SANKARA, UFR Sciences et Techniques, 12 BP 417 Ouagadougou 12, Ouagadougou, Burkina Faso.}

\address[rvt]{Universit\'e Bourgogne Franche-Comt\'e, Laboratoire de Math\'ematiques de Besan\c{c}on, UMR 6623 CNRS-UBFC, Besan\c{c}on, France}
\address[rvt1]{Universit\'e Joseph KI-ZERBO, LANIBIO, Ouagadougou, Burkina Faso}
\address[rvt2]{Universit\'e Thomas SANKARA, UFR ST, Ouagadougou, Burkina Faso}
\begin{abstract}

Multivariate nonnegative orthant data are real vectors bounded to the left by the null vector, and they can be continuous, discrete or mixed. We first review the recent relative variability indexes for multivariate nonnegative continuous and count distributions. As a prelude, the classification of two comparable distributions having the same mean vector is done through under-, equi- and over-variability with respect to the reference distribution. Multivariate associated kernel estimators are then reviewed with new proposals that can accommodate any nonnegative orthant dataset. We focus on bandwidth matrix selections by adaptive and local Bayesian methods for semicontinuous and counting supports, respectively. We finally introduce a flexible semiparametric approach for estimating all these  distributions on nonnegative supports. The corresponding estimator is directed by a given parametric part, and a nonparametric part which is a weight function to be estimated through multivariate associated kernels. A diagnostic model is also discussed to make an appropriate choice between the parametric, semiparametric and nonparametric approaches. The retention of pure nonparametric means the inconvenience of parametric part used in the modelization. Multivariate real data examples in semicontinuous setup as reliability are gradually considered to illustrate the proposed approach. Concluding remarks are made for extension to other multiple functions.
\end{abstract}

\end{frontmatter}

\section{Introduction}\label{1.Intro}

The $d$-variate nonnegative orthant data on $\mathbb{T}_d^+\subseteq [0,\infty)^d$ are real $d$-vectors bounded to the left by the null vector $\mathbf{0}_d$, and they can be continuous, discrete (e.g., count, categorial) or mixed. For simplicity, we here assume either $\mathbb{T}_d^+=[0,\infty)^d$ for semicontinuous or $\mathbb{T}_d^+=\mathbb{N}^d:=\{0,1,2,\ldots\}^d$ for counting; and, we then omit both setups of categorial and mixed which can be a mix of discrete and continuous data (e.g., \cite{SKI16}) or other {\em time scales} (see, e.g., \cite{Libengue13}). Modelling such datasets of $\mathbb{T}_d^+$ need nonnegative orthant distributions which are generally not easy to handle in practical data analysis. The baseline parametric distribution (e.g., \cite{Johnson97,Kotz00}) for the analysis of nonnegative countinuous data is the
exponential distribution (e.g., in Reliability) and that of count data is the Poisson one.
However, there intrinsic assumptions of the two first moments are often not realistic for
many applications. The nonparametric topic of associated kernels, which is adaptable to any
support $\mathbb{T}_d^+$ of probability density or mass function (pdmf), is widely studied in very recent years. We can refer to \cite{Belaid16,Belaid18,Funke15,Hirukawa18,KS18,Ouimet20,SK16,SK20,ZAK16,Zougab18} for general results and more specific developments on associated kernel orthant distributions using classical cross-validation and Bayesian methods to select bandwidth matrices. Thus, a natural question of a flexible semiparametric modelling now arises for all these multivariate orthant datasets.

Indeed, we first need a review of the recent relative variability indexes for multivariate semicontinuous (\cite{KTS20}) and count (\cite{KP18}) distributions. The infinite number and complexity of multivariate parametric distributions require the study of different indexes for comparisons and discriminations between them. Simple classifications of two comparable distributions are done through under-, equi- and over-variability with respect to the reference
distribution. We refer to \cite{Weiss19} and references therein for univariate categorial data which does not yet have its multivariate version.  We then survey multivariate associated kernels that can accommodate any nonnegative orthant dataset. Most useful families shall be pointed out, mainly as a product of univariate associated kernels and including properties and constructions. Also, we shall focus on bandwidth matrix selections by Bayesian methods. Finally, we have to introduce a flexible semiparametric approach for estimating multivariate nonnegative orthant distribution. Following 
\cite{HjortG95} for classical kernels, the corresponding estimator shall be directed by a given parametric part, and a nonparametric part which is a weight function to be estimated through multivariate associated kernels. But what is the meaning of a diagnostic model to make an appropriate choice between the parametric, semiparametric and nonparametric approaches in this multivariate framework? Such a discussion is to highlight practical improvements on standard nonparametric methods for multivariate semicontinuous datasets, through the use of a reasonable parametric-start description. See, for instance, \cite{KSKB09,KZSK17,SKZK16} for univariate count datasets.

In this paper, the main goal is to introduce a family of semiparametric estimators with multivariate associated kernels for both semicontinuous and count data. They are meant to be flexible compromises between a grueling parametric and fuzzy nonparametric approaches. The rest of the paper is organized as follow. Section \ref{2.Indexes} presents a brief review of the relative variability indexes for multivariate nonnegative orthant distributions, by distinguishing the dispersion for counting and the variation for semicontinuous. Section \ref{3.Kernels} displays a short panoply of multivariate associated kernels which are useful for semicontinuous and for count datasets. Properties are reviewed with new proposals, including both appropriated Bayesian methods of bandwidths selections. In Section \ref{4.Semiparam}, we introduce the semiparametric kernel estimators with $d$-variate parametric start. We also investigate the corresponding diagnostic model. Section \ref{5.Applic} is devoted to numerical illustrations, especially for uni- and  multivariate semicontinuous datasets. In Section \ref{6.Concl}, we make some final remarks in order to extend to other multiple functions, as regression. Eventually, appendixes are exhibited for technical proofs and illustrations.

\section{Relative Variability Indexes for Orthant Distributions}\label{2.Indexes}

Let $\boldsymbol{X} = (X_1,\ldots,X_d)^{\top}$ be a nonnegative orthant $d$-variate random vector on $\mathbb{T}_d^+\subseteq [0,\infty)^d$, $d\geq 1$. We use the following notations: $\sqrt{\mathrm{var}\boldsymbol{X}}=(\sqrt{\mathrm{var}X_1},\ldots,\sqrt{\mathrm{var}X_d})^{\top}$
is the elementwise square root of the variance vector of $\boldsymbol{X}$; 
$\mathrm{diag}\sqrt{\mathrm{var}\boldsymbol{X}}=\mathrm{diag}_d (\sqrt{\mathrm{var}X_j})$
is the $d\times d$ diagonal matrix with diagonal entries $\sqrt{\mathrm{var}X_j}$ and $0$ elsewhere; and, 
$\mathrm{cov}\boldsymbol{X}= (\mathrm{cov}(X_i,X_j) )_{i,j\in \{1,\ldots,d\}}$
denotes the covariance matrix of $\boldsymbol{X}$ which is a
$d\times d$ symmetric matrix with entries $\mathrm{cov}(X_i,X_j)$ such that
$\mathrm{cov}(X_i,X_i)=\mathrm{var}X_i$ is the variance of $X_i$. Then, one has
\begin{equation}\label{CovRhoVar}
\mathrm{cov}\boldsymbol{X}= (\mathrm{diag}\!\sqrt{\mathrm{var}\boldsymbol{X}})( \boldsymbol{{\textcolor{black}{\rho}}_{\boldsymbol{X}}})(\mathrm{diag}\!\sqrt{\mathrm{var}\boldsymbol{X}}),
\end{equation}
where
$\boldsymbol{\rho}_{\boldsymbol{X}}=\boldsymbol{\rho}(\boldsymbol{X})$
is the correlation matrix of $\boldsymbol{X}$; see, e.g., \cite[Eq. 2-36]{JohnsonW07}. It is noteworthy that there are huge many multivariate distributions with exponential (resp. Poisson)  margins. Therefore, we denote a generic $d$-variate exponential distribution by $\mathscr{E}_{d}(\boldsymbol{\mu},\boldsymbol{\rho})$, given specific positive mean vector $\boldsymbol{\mu}^{-1}:=(\mu_1^{-1},\ldots,\mu_d^{-1})^{\top}$ and correlation matrix $\boldsymbol{\rho}=(\rho_{ij})_{i,j \in \{1, \ldots, d\}}$. Similarly, a generic $d$-variate Poisson  distribution is given by $\mathscr{P}_{d}(\boldsymbol{\mu},\boldsymbol{\rho})$, with  positive mean vector $\boldsymbol{\mu}:=(\mu_1,\ldots,\mu_d)^{\top}$ and correlation matrix $\boldsymbol{\rho}$. See, e.g., Appendix \ref{AppendixA} for more extensive exponential and Poisson models with possible behaviours in the negative correlation setup. The uncorrelated or independent $d$-variate exponential and Poisson will be written as $\mathscr{E}_d(\boldsymbol{\mu})$ and $\mathscr{P}_d(\boldsymbol{\mu})$, respectively, for 
$\boldsymbol{\rho}=\boldsymbol{I}_d$ the $d\times d$ unit matrix. Their respective $d$-variate probability density function (pdf) and probability mass function (pmf) are the product of $d$ univariate ones. 

According to \cite{KTS20} and following the recent univariate unification of the well-known (Fisher) dispersion and the (J\o rgensen) variation indexes by Tour\'e {\it et al.} \cite{TDGK20}, the relative variability index of $d$-variate nonnegative orthant distributions can be written  as follows. 
Let $\boldsymbol{X}$ and $\boldsymbol{Y}$ be two random vectors on the same support $\mathbb{T}_d^+\subseteq [0,\infty)^d$ and assume $\boldsymbol{m}:=\mathbb{E}\boldsymbol{X}=\mathbb{E}\boldsymbol{Y}$, $\boldsymbol{\Sigma}_{\boldsymbol{X}}:=\mathrm{cov}\boldsymbol{X}$ and $\mathbf{V}_{F_{\boldsymbol{Y}}}(\boldsymbol{m}):=\mathrm{cov}(\boldsymbol{Y})$ fixed, then the
\textit{relative variability index} of $\boldsymbol{X}$ with respect to $\boldsymbol{Y}$ is defined as the positive quantity 
\begin{equation}\label{RWI}
\mathrm{RWI}_{\boldsymbol{Y}}(\boldsymbol{X}):=\mathrm{tr}[\boldsymbol{\Sigma}_{\boldsymbol{X}}\mathbf{W}^+_{F_{\boldsymbol{Y}}}(\boldsymbol{m})]\gtreqqless 1,
\end{equation}
where ``$\mathrm{tr}(\cdot)$'' stands for the trace operator and $\mathbf{W}^+_{F_{\boldsymbol{Y}}}(\boldsymbol{m})$ is the unique Moore-Penrose inverse of the associated matrix $\mathbf{W}_{F_{\boldsymbol{Y}}}(\boldsymbol{m}):=[\mathbf{V}_{F_{\boldsymbol{Y}}}(\boldsymbol{m})]^{1/2}[\mathbf{V}_{F_{\boldsymbol{Y}}}(\boldsymbol{m})]^{\top/2}$ to $\mathbf{V}_{F_{\boldsymbol{Y}}}(\boldsymbol{m})$. From (\ref{RWI}), $\mathrm{RWI}_{\boldsymbol{Y}}(\boldsymbol{X})\gtreqqless 1$ means the \textbf{over-} (\textbf{equi-} and \textbf{under-variability}) of $\boldsymbol{X}$ compared to $\boldsymbol{Y}$ is realized if $\mathrm{RWI}_{\boldsymbol{Y}}(\boldsymbol{X})>1$ ($\mathrm{RWI}_{\boldsymbol{Y}}(\boldsymbol{X})=1$  and $\mathrm{RWI}_{\boldsymbol{Y}}(\boldsymbol{X})<1$, respectively). 

The expression (\ref{RWI}) of RWI does not appear to be very easy to handle in this general formulation on $\mathbb{T}_d^+\subseteq [0,\infty)^d$, even the empirical version and interpretations. We now detail both multivariate cases of counting and of semicontinous. Their corresponding empirical versions are given in \cite{KP18,KTS20}.

\subsection{Relative Dispersion Indexes for Count Distributions}

For $\mathbb{T}_d^+=\mathbb{N}^d$, let $\mathbf{W}_{F_{\boldsymbol{Y}}}(\boldsymbol{m})=\sqrt{\boldsymbol{m}}\sqrt{\boldsymbol{m}}^\top$ be the $d\times d$ matrix of rank 1. Then, $\boldsymbol{\Sigma}_{\boldsymbol{X}}\mathbf{W}^+_{F_{\boldsymbol{Y}}}(\boldsymbol{m})$ of  (\ref{RWI}) is also of rank 1 and has only one positive eigenvalue, denoted by
\begin{equation}\label{GDI-def}
\mathrm{GDI}(\boldsymbol{X}) :=\frac{\sqrt{\mathbb{E}\boldsymbol{X}}^\top\, ( \mathrm{cov}\boldsymbol{X})\sqrt{\mathbb{E}\boldsymbol{X}}}{\mathbb{E}\boldsymbol{X}^{\top}\mathbb{E}\boldsymbol{X}}\gtreqqless 1
\end{equation} 
and called {\em generalized dispersion index} of $\boldsymbol{X}$ compared to $\mathbf{Y}\sim\mathscr{P}_d(\mathbb{E}\mathbf{X})$ with $\mathbb{E}\mathbf{Y}=\mathbb{E}\mathbf{X}=\boldsymbol{m}$ (\cite{KP18}). For $d=1$, $\mathrm{GDI}(X_1)=\mathrm{var}X_1/\mathbb{E}X_1=\mathrm{DI}(X_1)$ is the (Fisher) dispersion index with respect to the Poisson distribution. To derive this interpretation of GDI, we successively  decompose the denominator of (\ref{GDI-def}) as 
\begin{equation}\label{EspYdiag}
\mathbb{E}\boldsymbol{X}^{\top}\mathbb{E}\boldsymbol{X}=\sqrt{\mathbb{E}\boldsymbol{X}}^{\top}\, ( \mathrm{diag}\mathbb{E}\boldsymbol{X}) \sqrt{\mathbb{E}\boldsymbol{X}}
=[(\mathrm{diag}\sqrt{\mathbb{E}\boldsymbol{X}})\!\sqrt{\mathbb{E}\boldsymbol{X}}]^{\top}(\boldsymbol{I}_d)[ (\mathrm{diag}\sqrt{\mathbb{E}\boldsymbol{X}})\!\sqrt{\mathbb{E}\boldsymbol{X}}]
\end{equation}
and the numerator of (\ref{GDI-def}) by using also (\ref{CovRhoVar}) as
$$
\sqrt{\mathbb{E}\boldsymbol{X}}^\top\, ( \mathrm{cov}\boldsymbol{X})\sqrt{\mathbb{E}\boldsymbol{X}}=
[(\mathrm{diag}\!\sqrt{\mathrm{var}\boldsymbol{X}})\sqrt{\mathbb{E}\boldsymbol{X}}]^{\top} ( \boldsymbol{\rho}_{\boldsymbol{X}})\,[(\mathrm{diag}\!\sqrt{\mathrm{var}\boldsymbol{X}}) \sqrt{\mathbb{E}\boldsymbol{X}}].
$$
Thus, $\mathrm{GDI}(\boldsymbol{X})$ makes it possible to compare the full variability of $\boldsymbol{X}$ (in the numerator) with respect to its expected uncorrelated Poissonian variability (in the denominator) which depends only on $\mathbb{E}\boldsymbol{X}$. In other words, the count random vector $\mathbf{X}$ is \textbf{over-} (\textbf{equi-} and \textbf{under-dispersed}) with respect to $\mathscr{P}_d(\mathbb{E}\mathbf{X})$  if $\mathrm{GDI}(\boldsymbol{X})>1$ ($\mathrm{GDI}(\boldsymbol{X})=1$ and $\mathrm{GDI}(\boldsymbol{X})<1$, respectively). This is a generalization in multivariate framework of the well-known (univariate) Fisher dispersion index by  \cite{KP18}. See, e.g., \cite{Arnold20,KP18} for illustrative examples. Also, we can modify $\mathrm{GDI}(\boldsymbol{X})$ to $\mathrm{MDI}(\boldsymbol{X})$, as {\em marginal dispersion index}, by replacing $\mathrm{cov}\boldsymbol{X}$ in (\ref{GDI-def}) with $\mathrm{diag}\!\sqrt{\mathrm{var}\boldsymbol{X}}$ to obtain dispersion information only coming from the margins of $\boldsymbol{X}$.

More generally, for two count random vectors $\boldsymbol{X}$ and $\boldsymbol{Y}$ on the same support $\mathbb{T}_d^+\subseteq\mathbb{N}^d$ with $\mathbb{E}\boldsymbol{X}=\mathbb{E}\boldsymbol{Y}$ and $\mathrm{GDI}(\boldsymbol{Y})>0$, the {\em relative dispersion index} is defined by
\begin{equation}\label{RDI}
\mathrm{RDI}_{\boldsymbol{Y}}(\boldsymbol{X}): =\frac{\mathrm{GDI}(\boldsymbol{X})}{\mathrm{GDI}(\boldsymbol{Y})} = 
\frac{[(\mathrm{diag}\!\sqrt{\mathrm{var}\boldsymbol{X}})\sqrt{\mathbb{E}\boldsymbol{X}}]^{\top} ( \boldsymbol{\rho}_{\boldsymbol{X}})\,[(\mathrm{diag}\!\sqrt{\mathrm{var}\boldsymbol{X}}) \sqrt{\mathbb{E}\boldsymbol{X}}]}
{[(\mathrm{diag}\!\sqrt{\mathrm{var}\boldsymbol{Y}})\sqrt{\mathbb{E}\boldsymbol{Y}}]^{\top} ( \boldsymbol{\rho}_{\boldsymbol{Y}})\,[(\mathrm{diag}\!\sqrt{\mathrm{var}\boldsymbol{Y}}) \sqrt{\mathbb{E}\boldsymbol{Y}}]}
\gtreqqless 1;
\end{equation}
i.e., the {\bf over-} (\textbf{equi-} and \textbf{under-dispersion}) of $\boldsymbol{X}$ compared to $\boldsymbol{Y}$ is realized if $\mathrm{GDI}(\boldsymbol{X})>\mathrm{GDI}(\boldsymbol{Y})$ ($\mathrm{GDI}(\boldsymbol{X})=\mathrm{GDI}(\boldsymbol{Y})$ and $\mathrm{GDI}(\boldsymbol{X})<\mathrm{GDI}(\boldsymbol{Y})$, respectively). Obviously, GDI is a particular case of RDI with any general reference than $\mathscr{P}_d$. Consequently, many properties of GDI are easily extended to RDI.

\subsection{Relative Variation Indexes for Semicontinuous Distributions}

Assuming here $\mathbb{T}_d^+=[0,\infty)^d$ and $\mathbf{W}_{F_{\boldsymbol{Y}}}(\boldsymbol{m})=\boldsymbol{m}\boldsymbol{m}^\top$ another $d\times d$ matrix of rank 1. Then, we also have that $\boldsymbol{\Sigma}_{\boldsymbol{X}}\mathbf{W}^+_{F_{\boldsymbol{Y}}}(\boldsymbol{m})$ of  (\ref{RWI}) is of rank 1. Similar to (\ref{GDI-def}), the {\em generalized variation index} of $\boldsymbol{X}$ compared to $\mathscr{E}_d(\mathbb{E}\mathbf{X})$ is defined by
\begin{equation}\label{GVI-def}
\mathrm{GVI}(\boldsymbol{X}):=\frac{\mathbb{E}\boldsymbol{X}^\top\, ( \mathrm{cov}\boldsymbol{X})\;\mathbb{E}\boldsymbol{X}}{(\mathbb{E}\boldsymbol{X}^{\top}\mathbb{E}\boldsymbol{X})^2}\gtreqqless 1;
\end{equation}
i.e., $\mathbf{X}$ is \textbf{over-} (\textbf{equi-} and \textbf{under-varied}) with respect to $\mathscr{E}_d(\mathbb{E}\mathbf{X})$  if $\mathrm{GVI}(\boldsymbol{X})>1$ ($\mathrm{GVI}(\boldsymbol{X})=1$ and $\mathrm{GVI}(\boldsymbol{X})<1$, respectively); see  \cite{KTS20}.
Remark that when $d=1$, $\mathrm{GVI}(X_1)=\mathrm{var}X_1/(\mathbb{E}X_1)^2=\mathrm{VI}(X_1)$ is the univariate (J\o rgensen) variation index which is recently introduced by Abid {\it et al.} \cite{AKM20}. From (\ref{EspYdiag}) and using again (\ref{CovRhoVar}) for rewritting the numerator of (\ref{GVI-def}) as
$$
\mathbb{E}\boldsymbol{X}^\top\, ( \mathrm{cov}\boldsymbol{X})\;\mathbb{E}\boldsymbol{X}=
[(\mathrm{diag}\!\sqrt{\mathrm{var}\boldsymbol{X}})\mathbb{E}\boldsymbol{X}]^{\top} ( \boldsymbol{\rho}_{\boldsymbol{X}})\,[(\mathrm{diag}\!\sqrt{\mathrm{var}\boldsymbol{X}}) \mathbb{E}\boldsymbol{X}],
$$
$\mathrm{GVI}(\boldsymbol{X})$ of (\ref{GVI-def}) can be interpreted as the ratio of the full variability of $\boldsymbol{X}$ with respect to its expected uncorrelated exponential $\mathscr{E}_d(\mathbb{E}\mathbf{X})$ variability which depends only on $\mathbb{E}\boldsymbol{X}$. Similar to $\mathrm{MDI}(\boldsymbol{X})$, we can define $\mathrm{MVI}(\boldsymbol{X})$ from $\mathrm{GVI}(\boldsymbol{X})$. See \cite{KTS20} for properties, numerous examples and numerical illustrations.

The \textit{relative variation index} is defined, for two semicontinuous random vectors $\boldsymbol{X}$ and $\boldsymbol{Y}$ on the same support $\mathbb{T}_d^+=[0,\infty)^d$ with $\mathbb{E}\boldsymbol{X}=\mathbb{E}\boldsymbol{Y}$ and $\mathrm{GVI}(\boldsymbol{Y})>0$, by
\begin{equation}\label{RVI}
\mathrm{RVI}_{\boldsymbol{Y}}(\boldsymbol{X}): =\frac{\mathrm{GVI}(\boldsymbol{X})}{\mathrm{GVI}(\boldsymbol{Y})} = 
\frac{[(\mathrm{diag}\!\sqrt{\mathrm{var}\boldsymbol{X}})\mathbb{E}\boldsymbol{X}]^{\top} ( \boldsymbol{\rho}_{\boldsymbol{X}})\,[(\mathrm{diag}\!\sqrt{\mathrm{var}\boldsymbol{X}}) \mathbb{E}\boldsymbol{X}]}
{[(\mathrm{diag}\!\sqrt{\mathrm{var}\boldsymbol{Y}})\mathbb{E}\boldsymbol{Y}]^{\top} ( \boldsymbol{\rho}_{\boldsymbol{Y}})\,[(\mathrm{diag}\!\sqrt{\mathrm{var}\boldsymbol{Y}}) \mathbb{E}\boldsymbol{Y}]} 
\gtreqqless 1;
\end{equation}
i.e., the \textbf{over-} (\textbf{equi-} and \textbf{under-variation}) of $\boldsymbol{X}$ compared to $\boldsymbol{Y}$ is carried out if $\mathrm{GVI}(\boldsymbol{X})>\mathrm{GVI}(\boldsymbol{Y})$ ($\mathrm{GVI}(\boldsymbol{X})=\mathrm{GVI}(\boldsymbol{Y})$ and $\mathrm{GVI}(\boldsymbol{X})<\mathrm{GVI}(\boldsymbol{Y})$, respectively). Of course, RVI generalizes GVI for multivariate semicontinuous distributions. For instance, one refers to \cite{KTS20} for more details on its discriminating power in multivariate parametric models from two first moments.

\section{Multivariate Orthant Associated Kernels}\label{3.Kernels}

Nonparametric techniques through associated kernels represent an alternative approach for multivariate orthant data. Let $\textbf{X}_{1},\ldots,\textbf{X}_{n}$ be independent and identically distributed (iid) nonnegative orthant $d$-variate random vectors with an unknown joint pdmf $f$ on $\mathbb{T}_d^+\subseteq [0,\infty)^d$, for $d \geq 1$. Then the multivariate associated kernel estimator $\widetilde{f}_{n}$ of $f$ is expressed as
\begin{equation}\label{NPE}
\widetilde{f}_{n}(\mathbf{x})=\frac{1}{n} \sum_{i=1}^{n} \mathbf{K}_{\mathbf{x},\mathbf{H}}(\mathbf{X}_i),~~~\forall\mathbf{x}=(x_{1},\ldots, x_{d})^\top\in\mathbb{T}_d^+,
\end{equation}
where $\mathbf{H}$ is a given $d\times d$ bandwidth matrix (i.e., symmetric and positive definite) such that $\mathbf{H} \equiv \mathbf{H}_n \rightarrow \mathbf{0}_\mathbf{d}$ (the $d\times d$ null matrix) as $n\rightarrow\infty$, and $K_{\mathbf{x},\mathbf{H}}(\cdot)$ is a multivariate (orthant) associated kernel, parameterized by $\mathbf{x}$ and $\mathbf{H}$; see, e.g., \cite{KS18}. More precisely, we have the following refined definition.

\begin{Def}\label{def_MDAK}
	Let $\mathbb{T}_d^+$ be the support of the pdmf to be estimated, $\mathbf{x} \in \mathbb{T}_d^+$ a target vector and $\textbf{H}$ a bandwidth matrix. A parameterized pdmf $\mathbf{K}_{\mathbf{x},\mathbf{H}}(\cdot)$ on support $\mathbb{S}_{\mathbf{x},\mathbf{H}}\subseteq \mathbb{T}_d^+$ is called "multivariate orthant associated kernel" if the following conditions are satisfied:
	\begin{equation*}
	\mathbf{x}\in\mathbb{S}_{\mathbf{x},\mathbf{H}},\;
	\mathbb{E}\mathcal{Z}_{\mathbf{x},\mathbf{H}}=\mathbf{x}+\mathbf{A}(\mathbf{x},\mathbf{H})\to\mathbf{x} \;and\;
	\mathrm{cov}\mathcal{Z}_{\mathbf{x},\mathbf{H}}=\mathbf{B}(\mathbf{x},\mathbf{H})\to \mathbf{0}_{d}^{+},
	\end{equation*}
	where $\mathcal{Z}_{\mathbf{x},\mathbf{H}}$ denotes the corresponding orthant random vector with pdmf $\mathbf{K}_{\mathbf{x},\mathbf{H}}$ such that vector $\mathbf{A}(\mathbf{x},\mathbf{H})\to\mathbf{0}$ (the $d$-dimentional null vector) and positive definite matrix $\mathbf{B}(\mathbf{x},\mathbf{H})\to \mathbf{0}_{d}^{+}$ as $\mathbf{H}\to\mathbf{0}_d$ (the $d\times d$ null matrix), and $\mathbf{0}_{d}^{+}$ stands for a symmetric matrix with entries $u_{ij}$ for $i,j=1,\dots,d$ such that $u_{ij} \in [0,1)$.
\end{Def}

This definition exists in the univariate count case of \cite{KSK11,KZSK17} and encompasses the multivariate one by \cite{KS18}. The choice of the orthant associated kernel satisfying $\lim\limits_{\mathbf{H}\rightarrow \mathbf{0}_\mathbf{d}}\mathrm{Cov}(\mathcal{Z}_{\mathbf{x},\mathbf{H}})=\mathbf{0}_\mathbf{d}$ assures the convergence of its corresponding estimator named of the second order. Otherwise, the convergence of its corresponding estimator is not guarantee for $u_{ij}\in (0,1)$, a right neighborhood of $0$, in Definition \ref{def_MDAK} and it is said a consistent first-order smoother; see, e.g., \cite{KSK11} for discrete kernels. In general, $d$-under-dispersed count associated kernels are appropriated for both small and moderate sample sizes; see, e.g., \cite{KSK11} for univariate cases. As for the selection of the bandwidth $\mathbf{H}$, it is very crucial because it controls the degree of smoothing and the form of orientation of the kernel. As a matter of fact, a simplification can be obtained by considering a diagonal matrix $\mathbf{H}=\mathbf{diag}_d(h_j)$. Since it is challenging to obtain a full multivariate orthant distribution $\mathbf{K}_{\mathbf{x},\mathbf{H}}(\cdot)$ for building a smoother, several authors suggest the product of univariate orthant associated kernels,
\begin{equation}\label{KerProd}
\mathbf{K}_{\mathbf{x},\mathbf{H}}(\cdot)=\prod_{j=1}^{d} K_{x_j,h_{j}}(\cdot),
\end{equation}
where $K_{x_j,h_{j}}$, $j=1,\ldots,d$, belong either to the same family or to different families of univariate orthant associated kernels. The below two subsections shall be devoted to summaries of discrete and semicontinuous univariate associated kernels.

Before showing some main properties of the associated kernel estimator (\ref{NPE}), let us recall that the family of $d$-variate classical (symmetric) kernels $\mathbf{K}$ on $\mathbb{S}_d\subseteq\mathbb{R}^d$ (e.g., \cite{Scott92,Silverman86,ZAK14}) can be also presented as (classical) associated kernels. Indeed, from (\ref{NPE}) and writting for instance 
$$
\mathbf{K}_{\mathbf{x},\mathbf{H}}(\cdot)=(\det\mathbf{H})^{-1/2}\mathbf{K}\left[\mathbf{H}^{-1/2}(\mathbf{x}-\cdot)\right]
$$
where ``$\det$'' is the determinant operator, one has $\mathbb{S}_{\mathbf{x},\mathbf{H}}=\mathbf{x}-\mathbf{H}^{-1/2}\mathbb{S}_d$, $\mathbf{A}(\mathbf{x},\mathbf{H})=\mathbf{0}$ and $\mathbf{B}(\mathbf{x},\mathbf{H})=\mathbf{H}^{1/2}\boldsymbol{I}_d\mathbf{H}^{1/2}=\mathbf{H}$. In general, one uses the classical (associated) kernels for smoothing continuous data or pdf   having support $\mathbb{T}_d=\mathbb{R}^d$.

The purely nonparametric estimator (\ref{NPE}) with multivariate associated kernel, $\widetilde{f}_n$ of  $f$, is generally defined up to the normalizing constant $C_n$. Several simulation studies (e.g., \cite[Table 3.1]{KS18}) are shown that $C_n=C_n(\mathbf{K},\mathbf{H})$ (depending on samples, associated kernels and bandwidths) is approximatively $1$. Without loss of generality, one here assumes $C_n=1$ as for all classical (associated) kernel estimators of pdf. 
The following proposition finally  proves its mean behaviour and variability through the integrated bias and integrated variance of $\widetilde{f}_n$, respectively. In what follows, let us denote by $\boldsymbol{\nu}$ the reference measure (Lebesgue or counting) on the nonnegative orthant set $\mathbb{T}_d^+$ and also on any set $\mathbb{T}_d\subseteq\mathbb{R}^d$.

\begin{Pro}\label{PropCst_n}
	Let $C_n:=\int_{\mathbb{T}_d}\widetilde{f}_{n}(\mathbf{x})\boldsymbol{\nu}(d\mathbf{x})=C_n(\mathbf{K},\mathbf{H})$. Then, for all $n\geq 1$:
	\begin{equation*}
	\mathbb{E}(C_{n})=1+\int_{\mathbb{T}_d} \mathrm{Bias}\{\widetilde{f}_n(\mathbf{x})\}\boldsymbol{\nu}(d\mathbf{x})
	\;\;\;and\;\;\;
	\mathrm{var}(C_{n})=\int_{\mathbb{T}_d} \mathrm{var}\{\widetilde{f}_n(\mathbf{x})\}\boldsymbol{\nu}(d\mathbf{x}).
	\end{equation*}
\end{Pro}
\begin{proof}
	Let $n\geq 1$. One successively has
	$$\mathbb{E}(C_{n})=\int_{\mathbb{T}_d}\left[f(\mathbf{x})+ \mathbb{E}\{\widetilde{f}_n(\mathbf{x})\}-f(\mathbf{x})\right]\boldsymbol{\nu}(d\mathbf{x})
	=\int_{\mathbb{T}_d}f(\mathbf{x})\boldsymbol{\nu}(d\mathbf{x})+ \int_{\mathbb{T}_d}\left[\mathbb{E}\{\widetilde{f}_n(\mathbf{x})\}-f(\mathbf{x})\right]\boldsymbol{\nu}(d\mathbf{x}),$$
	which leads to the first result because $f$ is a pdmf on $\mathbb{T}_d$. The second result on $\mathrm{var}(C_{n})$ is trivial. 
\end{proof}

The following general result is easily deduced from Proposition \ref{PropCst_n}. 
To our knwoledge, it appears to be new and interesting in the framework of the pdmf (associated) kernel estimators.

\begin{Cor}\label{Coroll}
	If $C_n=1$, for all $n\geq 1$, then: $\int_{\mathbb{T}_d} \mathrm{Bias}\{\widetilde{f}_n(\mathbf{x})\}\boldsymbol{\nu}(d\mathbf{x})=0$ and $\int_{\mathbb{T}_d} \mathrm{var}\{\widetilde{f}_n(\mathbf{x})\}\boldsymbol{\nu}(d\mathbf{x})=0$.
\end{Cor}

In particular, Corollary \ref{Coroll} holds for all classical (associated) kernel estimators. 
The two following properties on the corresponding orthant multivariate associated kernels shall be needed subsequently.

\begin{description}
	\item\label{K1} ({\bf K1}) There exists the second moment of $\mathbf{K}_{\mathbf{x},\mathbf{H}}$:
	$$\mu_j^2(\mathbf{K}_{\mathbf{x},\mathbf{H}}):=\int_{\mathbb{S}_{\mathbf{x},\mathbf{H}}\bigcap\mathbb{T}_d^+} u_j^2\mathbf{K}_{\mathbf{x},\mathbf{H}}(\mathbf{u})\boldsymbol{\nu}(d\mathbf{u})<\infty,\;\;\;\forall j=1,\ldots,d.$$
	\item\label{K2} ({\bf K2}) There exists a real largest number $r=r(\mathbf{K}_{\mathbf{x},\mathbf{H}})>0$ and $0<c(\mathbf{x})<\infty$ such that 
	$$||\mathbf{K}_{\mathbf{x},\mathbf{H}}||_2^2:=\int_{\mathbb{S}_{\mathbf{x},\mathbf{H}}\bigcap\mathbb{T}_d^+} \{\mathbf{K}_{\mathbf{x},\mathbf{H}}(\mathbf{u})\}^2\boldsymbol{\nu}(d\mathbf{u})\leq c(\mathbf{x})(\det\mathbf{H})^{-r}.
	$$
\end{description}
In fact, ({\bf K1}) is a necessary condition for smoothers to have a finite variance and ({\bf K2}) can be deduced from the continuous univariate cases (e.g., \cite{KL18}) and also from the discrete ones (e.g., \cite{KSK11}).

We now establish both general asymptotic behaviours of the  pointwise bias and variance of the nonparametric estimator (\ref{NPE}) on the  nonnegative orthant set $\mathbb{T}_d^+$; its proof is given in Appendix \ref{AppendixB}. For that, we need the following assumptions by endowing $\mathbb{T}_d^+$ with the Euclidean norm $||\cdot||$ and the associated inner product $\langle\cdot,\cdot\rangle$ such that $\langle\mathbf{a},\mathbf{b}\rangle=\mathbf{a}^\top\mathbf{b}$. 
\begin{description}
	\item\label{a1} ({\bf a1}) The unknown pdmf $f$ is bounded function and twice differentiable or finite difference in $\mathbb{T}_d^+$ and $\nabla f(\mathbf{x})$ and $\mathcal{H}f(\mathbf{x})$ denote respectively the gradient vector (in continuous or discrete sense, respectively) and the corresponding Hessian matrix of the function $f$ at $\mathbf{x}$.
	
	\item\label{a2} ({\bf a2}) There exists a positive real  number $r>0$   such that $||K_{\mathbf{x},\mathbf{H}_n}||_2^2(\det\mathbf{H}_n)^{r}\to c_1(\mathbf{x})>0$ as $n\to\infty$.
\end{description}

Note that ({\bf a2}) is obviously a consequence of ({\bf K2}).

\begin{Pro}\label{PropBiasVarf(x)}
	Under the assumption ({\bf a1}) on $f$, then the estimator $\widetilde{f}_n$ in (\ref{NPE}) of $f$ verifies
	\begin{equation}\label{Biais} 
	\mathrm{Bias}\{\widetilde{f}_n(\mathbf{x})\}
	=\left\langle\nabla f(\mathbf{x}), \mathbf{A}\left(\mathbf{x}, \mathbf{H}_n\right)\right\rangle+
	\frac{1}{2} \operatorname{tr} \left\{{\cal H} f\left(\mathbf{x}\right)\left[\mathbf{B}(\mathbf{x},\mathbf{H}_n)+\mathbf{A}\left(\mathbf{x},\mathbf{H}_n\right)^\mathsf{T}\mathbf{A}\left(\mathbf{x},\mathbf{H}_n\right)\right]\right\}
	+o\left\{\operatorname{tr}\left[\mathbf{B}(\mathbf{x},\mathbf{H}_n)\right]\right\}, 
	\end{equation}  
	for any $\mathbf{x}\in\mathbb{T}_d^+$. Moreover, if ({\bf a2}) holds then
	\begin{equation}\label{Variance}
	\mathrm{var}\{\widetilde{f}_n(\mathbf{x})\}
	=\frac{1}{n}f(\mathbf{x})||\mathbf{K}_{\mathbf{x},\mathbf{H}_n}||_2^2+o\left[\frac{1}{n(\det\mathbf{H}_n)^r}\right].
	\end{equation}
\end{Pro}

For $d=1$ and according to the proof of Proposition \ref{PropBiasVarf(x)}, one can easily write $\mathbb{E}\widehat{f}_n(x)$ as follows:
$$
\mathbb{E}\widehat{f}_n(x)=\mathbb{E}f(\mathcal{Z}_{x,h})=\sum_{k\geq 0}\frac{1}{k!}\mathbb{E}\left(\mathcal{Z}_{x,h}-\mathbb{E}\mathcal{Z}_{x,h}\right)^kf^{(k)}(\mathbb{E}\mathcal{Z}_{x,h}),
$$
where $f^{(k)}$ is the $k$th derivative or finite difference of the pdmf $f$ under the existence of the centered moment of order $k\geq 2$ of $\mathcal{Z}_{x,h}$.

Concerning bandwidth matrix selections in a multivariate associated kernel estimator (\ref{NPE}), one generally use the cross-validation technique (e.g., \cite{KSK11,KSKB09,KSKZ07,KS18,WSK16}). However, it is tedious and less precise. Many papers have recently proposed Bayesian approaches (e.g., \cite{Belaid16,Belaid18,Some20,SK20,Ziane15,ZAK16}  and references therein). In particular, they have recommended local Bayesian for discrete smoothing of pmf (e.g., \cite{Belaid16,Belaid18,Belaid20}) and adaptive one for continuous smoothing of pdf (e.g., \cite{Some20,SK20,Ziane15}).

Denote $\mathcal{M}$ the set of positive definite [diagonal]  matrices [from (\ref{KerProd}), resp.] and let $\pi$ be a given suitable prior distribution on $\mathcal{M}$. 
Under the squared error loss function, the Bayes estimator of $\mathbf{H}$ is the mean of the posterior distribution. Then, the local Bayesian bandwidth at the target $\mathbf{x}\in\mathbb{T}_d^+$ takes the form 
\begin{equation}\label{BayesLoc}
\widetilde{\mathbf{H}}(\mathbf{x}):=\int_{\mathcal{M}}\mathbf{H}\pi(\mathbf{H})\widetilde{f}_{n}(\mathbf{x})d\mathbf{H}\left[\int_{\mathcal{M}}\widetilde{f}_{n}(\mathbf{x})\pi(\mathbf{H})d\mathbf{H}\right]^{-1},\;\;\mathbf{x}\in\mathbb{T}_d^+,
\end{equation}
and the adaptive Bayesian bandwidth for each observation $\mathbf{X}_i\in\mathbb{T}_d^+$ of $\mathbf{X}$ is given by
\begin{equation}\label{BayesAdap}
\widetilde{\mathbf{H}}_i:=\int_{\mathcal{M}}\mathbf{H}_i\pi(\mathbf{H}_i)\widetilde{f}_{n,\mathbf{H}_i,-i}(\mathbf{X}_i)d\mathbf{H}_i\left[\int_{\mathcal{M}}\widetilde{f}_{n,\mathbf{H}_i,-i}(\mathbf{X}_i)\pi(\mathbf{H}_i)d\mathbf{H}_i\right]^{-1},\;\;i=1,\ldots,n,
\end{equation}
where $\widetilde{f}_{n,\mathbf{H}_i,-i}(\mathbf{X}_i)$ is the leave-one-out associated kernel estimator of $f(\mathbf{X}_i)$ deduced from (\ref{NPE}) as
\begin{equation}\label{NPE_loo}
\widetilde{f}_{n,\mathbf{H}_i,-i}(\mathbf{X}_i):=
\frac{1}{n-1}\sum_{\ell=1,\ell\neq i}^{n} \mathbf{K}_{\mathbf{X}_i,\mathbf{H}_i}(\mathbf{X}_\ell).
\end{equation}
Note that the (global) cross-validation bandwidth matrix $\widetilde{\mathbf{H}}_{CV}$ and the global Bayesian one $\widetilde{\mathbf{H}}_{B}$ are obtained, respectively, from (\ref{NPE_loo}) as 
$$\widetilde{\mathbf{H}}_{CV}:=\mathrm{arg}\min_{\mathbf{H}\in\mathcal{M}}\left[\int_{\mathbb{T}_d^+}\{\widetilde{f}_{n}(\mathbf{x})\}^2\boldsymbol{\nu}(d\mathbf{x})-\frac{2}{n}\sum_{i=1}^n\widetilde{f}_{n,\mathbf{H},-i}(\mathbf{X}_i)\right]
$$
and 
\begin{equation*}\label{BayesGlob}
\widetilde{\mathbf{H}}_B:=\int_{\mathcal{M}}\mathbf{H}\pi(\mathbf{H})d\mathbf{H}\prod_{i=1}^n\widetilde{f}_{n,\mathbf{H},-i}(\mathbf{X}_i)d\mathbf{H}\left[\int_{\mathcal{M}}\pi(\mathbf{H})d\mathbf{H}\prod_{i=1}^n\widetilde{f}_{n,\mathbf{H},-i}(\mathbf{X}_i)\right]^{-1}.
\end{equation*}

\subsection{Discrete Associated Kernels}

We only present three main and useful families of univariate  discrete associated kernels for (\ref{KerProd}) and satisfying ({\bf K1}) and ({\bf K2}). 

\begin{Ex}[categorical]\label{ExCateg}
	For fixed $c\in\{2,3,\ldots\}$, the number of categories and $\mathbb{T}_1^+=\{0,1,\ldots,c-1\}$, one defines the Dirac discrete uniform (DirDU) kernel by
	$$K_{x,h}^{DirDU}(u)=(1-h)^{\mathds{1}_{u=x}}\left(\frac{h}{c-1}\right)^{1-\mathds{1}_{u=x}},$$ for $x\in\{0,1,\ldots,c-1\}$, $h\in(0,1]$, 
	with $\mathbb{S}_x:=\{0,1,\ldots,c-1\}=\mathbb{T}_1^+$, $A(x,h)=h\{c/2-x-x/(c-1)\}$ and $B(x,h)=h\{c(2c-1)/6+x^2-xc+x^2/(c-1)\}-h^2\{c/2-x-x/(c-1)\}^2$.
\end{Ex}
It has been introduced in multivariate setup by Aitchison and Aitken \cite{Aitchison76} and investigated as a discrete associated kernel which is symmetric to the target $x$ by \cite{KSK11} in univariate case; see \cite{Belaid18} for a Bayesian approach in multivariate setup. Note here that its normalized constant is always $1=C_n$. 

\begin{Ex}[symmetric count]\label{ExCount}
	For fixed $m\in\mathbb{N}$ and $\mathbb{T}_1^+\subseteq\mathbb{Z}$, the symmetric count triangular kernel is expressed as
	$$K_{x,h}^{SCTriang}(u)=\frac{(m+1)^h-|u-x|^h}{P(m,h)}\mathds{1}_{\{x,x\pm 1,\ldots,x\pm m\}}(u),$$
	for $x\in\mathbb{T}_1^+$, $h>0$, 
	with $\mathbb{S}_x:=\{x,x\pm 1,\ldots,x\pm m\}$, $P(m,h)=(2m+1)(m+1)^h-2\sum_{\ell=0}^m\ell^h$,  $A(x,h)=0$ and
	\begin{eqnarray*}
		B(x,h) & = & \frac{1}{P(m,h)}\left\{\frac{m(2m+1)(m+1)^{h+1}}{3}-2\sum_{\ell=0}^m\ell^{h+2}\right\}\\
		& \simeq & h\left\{\frac{m(2m^2+3m+1)}{3}\log(m+1)-2\sum_{\ell=1}^m\ell^{2}\log\ell\right\}+O(h^2),
	\end{eqnarray*}
	where $\simeq$ holds for $h$ sufficiently small.
\end{Ex}
It has been first proposed by Kokonendji {\it et al.} \cite{KSKZ07} and then completed in \cite{KZ10} with an asymmetric version for solving the problem of boundary bias in count kernel estimation.

\begin{Ex}[standard count]\label{ExCountSt}
	Let $\mathbb{T}_1^+\subseteq\mathbb{N}$, the standard binomial kernel is defined by
	$$K_{x,h}^{Binomial}(u)=\frac{(x+1)!}{u!(x+1-u)!}\left(\frac{x+h}{x+1}\right)^u\left(\frac{1-h}{x+1}\right)^{x+1-u}\mathds{1}_{\{0,1,\ldots,x+1\}}(u),$$
	for $x\in\mathbb{T}_1^+$, $h\in(0,1]$, 
	with $\mathbb{S}_x:=\{0,1,\ldots,x+1\}$, $A(x,h)=h$ and $B(x,h)=(x+h)(1-h)/(x+1)\to x/(x+1)\in [0,1]$ as $h\to 0$.
\end{Ex}

Here $B(x,h)$ tends to $x/(x+1)\in [0,1)$ when $h\to 0$ and the new Definition \ref{def_MDAK} holds. This first-order and under-dispersed  binomial kernel is introduced in \cite{KSK11} which becomes very useful for smoothing count distribution through small or moderate sample size; see, e.g., \cite{Belaid16,Belaid18,Belaid20} for Bayesian approaches and some references therein. In addition, we have the standard Poisson kernel where $K_{x,h}^{Poisson}$ follows the equi-dispersed Poisson distribution with mean $x+h$, $\mathbb{S}_x:=\mathbb{N}=:\mathbb{T}_1^+$, $A(x,h)=h$ and $B(x,h)=x+h\to x\in\mathbb{N}$ as $h\to 0$. Recently, Huang {\it et al.} \cite{Huang20} have introduced the Conway-Maxwell-Poisson kernel by exploiting its under-dispersed part and its second-order consistency which can be improved via the mode-dispersion approach of \cite{LK17}; see also \cite[Section 2.4]{HuangA20}.

\subsection{Semicontinuous Associated Kernels}

Now, we point out eight main and useful families of univariate semicontinuous associated kernels  for (\ref{KerProd}) and satisfying ({\bf K1}) and ({\bf K2}). Which are gamma (G) of \cite{Chen00} (see also \cite{Hirukawa15}), inverse gamma (Ig) (see also \cite{Mousa16}) and log-normal 2 (LN2) by \cite{LK17}, inverse Gaussian (IG) and reciprocal inverse Gaussian by \cite{Scaillet04} (see also \cite{Igarashi14}), log-normal 1 (LN1) and Birnbaum-Saunders by \cite{Jin03} (see also \cite{Marchant13,Mombeni19}), and Weibull (W) of \cite{Salha14} (see also \cite{Mombeni19}). It is noteworthy that the link between LN2 of \cite{LK17} and LN1 of \cite{Jin03} is through changing $(x,h)$ to $(x\exp(h^2),2\sqrt{\log(1+h})$. Several other semicontinuous could be constructed by using the  mode-dispersion technique of \cite{LK17} from any semicontinuous distribution which is unimodal and having a dispersion parameter. Recently, one has the scaled inverse chi-squared kernel of \cite{Ercelik20}.

Table \ref{Table_1sAK} summarizes these eight semicontinuous univariate associated kernels with their ingredients of Definition \ref{def_MDAK} and order of preference (O.) obtained graphically. In fact, the heuristic classification (O.) is done through the behaviour of the shape and scale of the associated kernel around the target $x$ at the edge as well as inside; see Figure \ref{fig_03ab} for edge and Figure \ref{fig_23ab} for inside. Among these eight kernels, we thus have to recommend the five first univariate associated kernels of Table \ref{Table_1sAK} for smoothing semicontinuous data. This approach could be improved by a given dataset; see, e.g.,  \cite{Lafaye20} for cumulative functions.

\section{Semiparametric Kernel Estimation with $d$-Variate Parametric Start}\label{4.Semiparam}

We investigate the semiparametric orthant kernel approach which is a compromise between the pure parametric and the nonparametric methods. This concept was proposed by \citet{HjortG95} for continuous data, treated by \citet{KSKB09} for discrete univariate data and, recently, studied by \citet{KZSK17} with an application to radiation biodosimetry. 

\begin{landscape}
	\begin{table}[htp]
	\vspace{3.6cm}
		\caption{Eight semicontinuous univariate associated kernels on $\mathbb{S}_{x,h}\subseteq [0,\infty)$ and classifyed by "O."}\label{Table_1sAK}
		\begin{tabular}{lllll}
			\toprule
			\textbf{O.}	& \textbf{Name}	 & \textbf{$K_{x,h}(u)$} & \textbf{$A(x,h)$} & \textbf{$B(x,h)$} \\
			\midrule
			1 & LN2 \cite{LK17} & $(uh\sqrt{2\pi})^{-1}\exp\left(\left[\log\{x\exp(h^2)\}-\log u\right]/2h^2\right)$ & $x[\exp(3h^2\!/2)\!-\!1]$ & $x^2\exp(3h^2)[\exp(h^2)-1]$ \\
			2 & W \cite{Salha14} & $[\Gamma(h)/x][u\Gamma(1+h)/x]^{1/h-1}\exp\left\{-[u\Gamma(1+h)/x]^{1/h}\right\}$ & $0$ & $x^2\!\left[\Gamma(1\!+\!2h)/\Gamma^2(1\!+\!2h)\!-\!1\right]$ \\
			3 & G \cite{Chen00} & $h^{-1-x/h}u^{x/h}\exp(-u/h)/\Gamma(1+x/h)$ & $h$ & $(x+h)h$ \\
			4 & BS \cite{Jin03} & $(uh\sqrt{2\pi})^{-1}\!\left[(xu)^{-1/2}\!+\!(x/u^3)^{-1/2}\right]\exp\left[(2\!-\!u/x\!-\!x/u)/2h\right]$ & $xh/2$ & $x^2h(2+5h/2)/2$ \\
			5 & Ig \cite{LK17} & $h^{1-1/xh}u^{-1/xh}\exp(-1/uh)/\Gamma(1/xh -1)$ & $2x^2h/(1-2xh)$ & $x^3h/[(1-3xh)(1-2xh)^2]$ \\
			6 & RIG \cite{Scaillet04} & $(\sqrt{2\pi uh})^{-1}\exp\left\{[x-h][2-(x-h)/u-u/(x-h)]/2h\right\}$ &$0$ & $(x-h)h$ \\
			7 & IG \cite{Scaillet04} & $(\sqrt{2\pi hu^3})^{-1}\exp\left\{[2-x/u-u/x)]/2hx\right\}$ & $0$ & $x^3h$\\
			8 & LN1 \cite{Jin03} & $(u\sqrt{8\pi\log(1+h)})^{-1}\exp\left(-[\log u-\log x]^2/[8\log(1+h)]\right)$ & $xh(h+2)$ & $x^2(1+h)^4[(1+h)^4-1]$ \\
			\bottomrule\noalign{\smallskip}
		\end{tabular}
		$\Gamma(v):=\int_0^\infty s^{v-1}\exp(-s)ds$ is the classical gamma function with $v>0$. 
	\end{table}
\end{landscape}

\begin{figure}[htp]
	\centering
	\includegraphics[width=7.5cm,height=6cm]{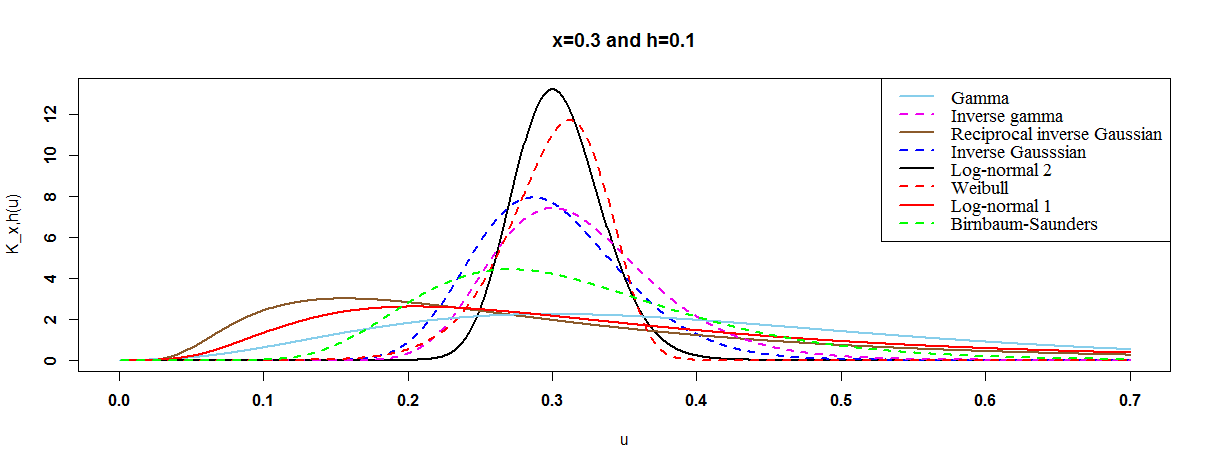}
	\includegraphics[width=7.5cm,height=6cm]{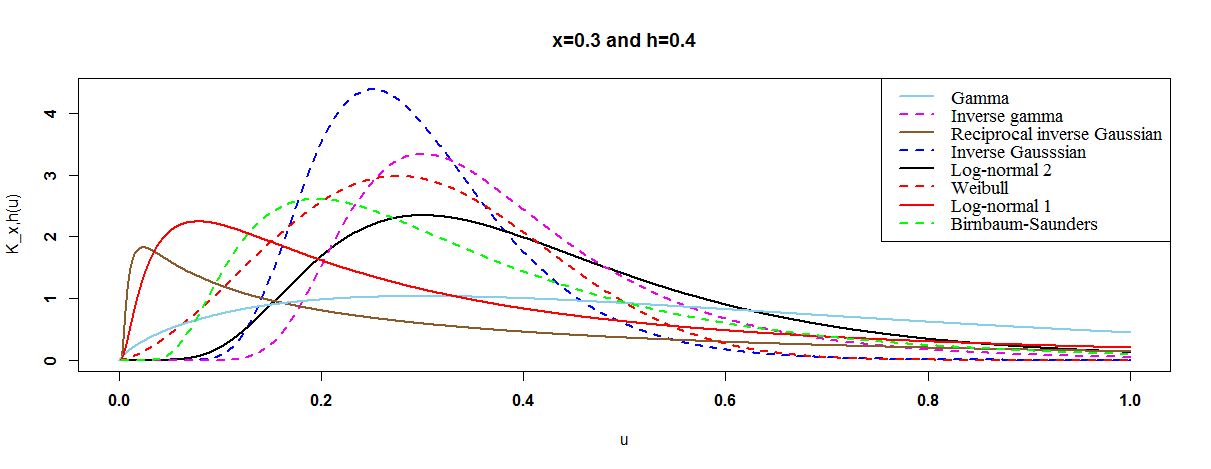}
	\caption{Comparative graphics of the eight univariate semicontinuous associated kernels of Table \ref{Table_1sAK} on the edge ($x=0.3$) with $h=0.1$ and $h=0.4$.}\label{fig_03ab}
\end{figure} 
\begin{figure}[H]
	\centering
	\includegraphics[width=7.5cm,height=6cm]{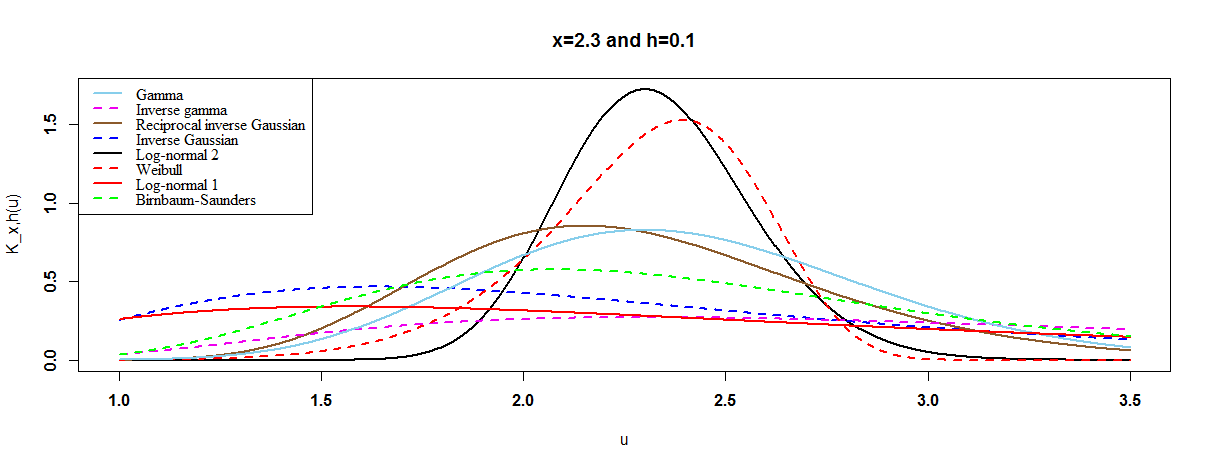}
	\includegraphics[width=7.5cm,height=6cm]{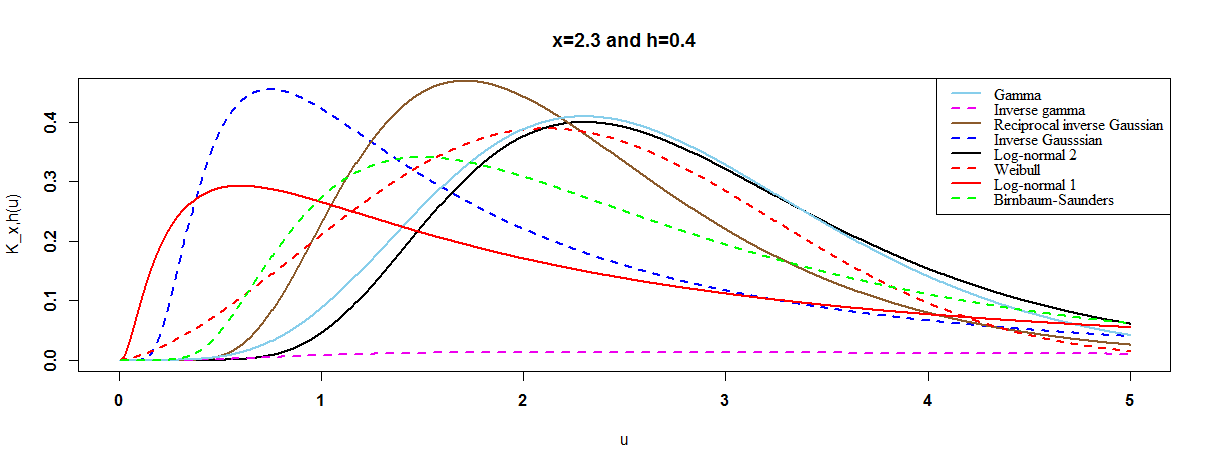}
	\caption{Comparative graphics of the eight univariate semicontinuous associated kernels of Table \ref{Table_1sAK} inside ($x=2.3$) with $h=0.1$ and $h=0.4$.}\label{fig_23ab}
\end{figure}     

Without loss of generality, we here assume that any $d$-variate pdmf $f$ can be formulated (e.g., \cite{KTA20} for $d=1$) as
\begin{equation}\label{semiP}
f(\mathbf{x})=w(\mathbf{x};\boldsymbol{\theta})\,p_{d}(\mathbf{x};\boldsymbol{\theta}),\;\;\; \forall \mathbf{x}\in \mathbb{T}_d^+,
\end{equation}
where $p_d(\cdot;\boldsymbol{\theta})$ is the non-singular parametric part according to a reference $d$-variate distribution with corresponding unknown  parameters $\boldsymbol{\theta}=(\theta_1,\ldots, \theta_k)^\top$ and $w(\cdot;\boldsymbol{\theta}):=f(\cdot)/p_d(\cdot;\boldsymbol{\theta})$ is the unknown orthant weight function part, to be estimated with a multivariate orthant associated kernel. The weight function at each point can be considered as the local multiplicative correction factor aimed to accommodate any pointwise departure from the reference $d$-variate distribution. However, one cannot consider the best fit of parametric models as the start distribution in this semiparametric approach. Because the corresponding weight function is close to zero and becomes a noise  which is unappropriated to smooth by an associated kernel, especially for the continuous cases.

Let $\mathbf{X}_1,\ldots,\mathbf{X}_n$ be iid  nonnegative orthant $d$-variate random vectors with unknown pdmf $f$ on $\mathbb{T}_{d}^+\subseteq [0,\infty)^d$. The semiparametric estimator of (\ref{semiP}) with (\ref{KerProd}) is expressed as follows:
\begin{eqnarray}\label{SME}
\widehat{f}_{n}(\mathbf{x})&=& p_{d}(\mathbf{x};\widehat{\boldsymbol{\theta}}_n)\frac{1}{n}\sum_{i=1}^{n}\frac{1}{p_{d}(\mathbf{X}_{i};\widehat{\boldsymbol{\theta}}_n)}\mathbf{K}_{\mathbf{x},\mathbf{H}}(\mathbf{X}_i)\nonumber\\
&=& \frac{1}{n} \sum_{i=1}^{n}\frac{p_{d}(\mathbf{x};\widehat{\boldsymbol{\theta}}_n)}{p_{d}(\mathbf{X}_{i};\widehat{\boldsymbol{\theta}}_n)}\mathbf{K}_{\mathbf{x},\mathbf{H}}(\mathbf{X}_i),\;\;\;\mathbf{x}\in\mathbb{T}_d^+, 
\end{eqnarray}
where $\widehat{\boldsymbol{\theta}}_n$ is the estimated parameter of $\boldsymbol{\theta}$.  From (\ref{SME}), we then deduce the nonparametric orthant associated kernel estimate
\begin{equation}\label{w_tilde}
\widetilde{w}_{n}(\mathbf{x};\widehat{\boldsymbol{\theta}}_n)= \frac{1}{n}\sum_{i=1}^{n}\frac{1}{p_{d}(\mathbf{X}_{i};\widehat{\boldsymbol{\theta}}_n)}\mathbf{K}_{\mathbf{x},\mathbf{H}}(\mathbf{X}_i)
\end{equation}
of the weight function $x \mapsto w(\mathbf{x};\widehat{\boldsymbol{\theta}_{n}})$ which depends  on $\widehat{\boldsymbol{\theta}}_n$. One can observe that Proposition \ref{PropCst_n} also holds for $\widehat{f}_{n}(\cdot)= p_{d}(\cdot ;\widehat{\boldsymbol{\theta}}_n) \widetilde{w}_{n}(\cdot ;\widehat{\boldsymbol{\theta}}_n)$. However, we have to prove below the analogous of Proposition \ref{PropBiasVarf(x)}.

\subsection{Known $d$-Variate Parametric Model}

Let $p_{d}(\cdot;\boldsymbol{\theta}_0)$ be a fixed orthant distribution in (\ref{semiP}) with $\boldsymbol{\theta}_0$ known. Writing $f(\mathbf{x})=p_{d}(\mathbf{x};\boldsymbol{\theta}_0)\,w(\mathbf{x})$, we estimate the nonparametric weight function $w$ by $\widetilde{w}_n(\mathbf{x})=n^{-1}\sum_{i=1}^n \mathbf{K}_{\mathbf{x},\mathbf{H}}(\mathbf{X}_i)/p_{d}(\mathbf{X}_i;\boldsymbol{\theta}_0)$ with an orthant associated kernel method, resulting in the estimator
\begin{equation}\label{MP01}
\widehat{f}_{n}(\mathbf{x})=p_{d}(\mathbf{x};\boldsymbol{\theta}_0)\widetilde{w}_n(\mathbf{x})=\frac{1}{n} \sum\limits_{i=1}^{n} \frac{p_{d}
	(\mathbf{x};\boldsymbol{\theta}_0)}{p_{d}(\mathbf{X}_{i};\boldsymbol{\theta}_0)}\mathbf{K}_{\mathbf{x},\mathbf{H}}(\mathbf{X}_i),\;\;\;\;\;\mathbf{x}\in\mathbb{T}_d^+.\end{equation}
The following proposition is proven in Appendix \ref{AppendixB}.

\begin{Pro}\label{PropBiasVarf(x,0)}
	Under the assumption ({\bf a1}) on $f(\cdot)=p_{d}(\cdot;\boldsymbol{\theta}_0)w(\cdot)$, then the estimator $\widehat{f}_n (\cdot)=p_{d}(\cdot;\boldsymbol{\theta}_0)\widetilde{w}_n(\cdot)$ in (\ref{MP01}) of $f$ satisfies
	\begin{eqnarray*}\label{Biaisf(x,0)} 
		\mathrm{Bias}\{\widehat{f}_n(\mathbf{x})\}&=&
		p_{d}(\mathbf{x};\boldsymbol{\theta}_0)\left[w(\mathbf{x})-f(\mathbf{x})\{p_{d}(\mathbf{x};\boldsymbol{\theta}_0)\}^{-1}+\left\langle\nabla w(\mathbf{x}),\mathbf{A}\left(\mathbf{x}, \mathbf{H}_n\right)\right\rangle\right] \\
		&&+
		\frac{1}{2}\; p_{d}(\mathbf{x};\boldsymbol{\theta}_0)\left(\operatorname{tr} \left\{{\cal H} w\left(\mathbf{x}\right)\left[\mathbf{B}(\mathbf{x},\mathbf{H}_n)+\mathbf{A}\left(\mathbf{x},\mathbf{H}_n\right)^\mathsf{T}\mathbf{A}\left(\mathbf{x},\mathbf{H}_n\right)\right]\right\}\right)\\
		&&+\left(1+o\left\{\operatorname{tr}\left[\mathbf{B}(\mathbf{x},\mathbf{H}_n)\right]\right\}\right),
	\end{eqnarray*}  
	for any $\mathbf{x}\in\mathbb{T}_d^+$.  Furthermore, if ({\bf a2}) holds then one has $\mathrm{var}\{\widehat{f}_n(\mathbf{x})\}
	=\mathrm{var}\{\widetilde{f}_n(\mathbf{x})\}$ of (\ref{Variance}).
\end{Pro}
It is expected that the bias here is quite different from that of (\ref{Biais}).

\subsection{Unknown $d$-Variate Parametric Model}

Let us now consider the more realistic and practical semiparametric estimator $\widehat{f}_{n}(\cdot)= p_{d}(\cdot ;\widehat{\boldsymbol{\theta}}_n) \widetilde{w}_{n}(\cdot ;\widehat{\boldsymbol{\theta}}_n)$ presented in (\ref{SME}) of $f(\cdot)= p_{d}(\cdot ;\boldsymbol{\theta}) w(\cdot ;\boldsymbol{\theta})$ in (\ref{semiP}) such that the parametric estimator $\widehat{\boldsymbol{\theta}}_n$ of $\boldsymbol{\theta}$ can be obtained by the maximum likelihood method; see \cite{HjortG95} for quite a general estimator of $\boldsymbol{\theta}$. In fact, if the $d$-variate parametric model $p_{d}(\cdot ;\boldsymbol{\theta})$ is misspecified then this  $\widehat{\boldsymbol{\theta}}_n$ converges in probability to the pseudotrue value $\boldsymbol{\theta}_0$ satisfying 
$$\boldsymbol{\theta}_0:=\arg\min\limits_{\boldsymbol{\theta}}\int_{\mathbf{x} \in \mathbb{T}_d^+} f(\mathbf{x})\log[f(\mathbf{x})/p_{d}(\mathbf{x};\boldsymbol{\theta})]\boldsymbol{\nu}(d\mathbf{x})$$ 
from the Kullback-Leibler divergence (see, e.g., \cite{White82}).

By writting $p_0(\cdot):=p_{d}(\cdot ;\boldsymbol{\theta}_0)$ this best $d$-variate parametric approximant, but this $p_0(\cdot)$ is not explicitly expressible as the one in (\ref{MP01}). According to \cite{HjortG95} (see also \cite{KSKB09}), we can represent the proposed estimator $\widehat{f}_{n}(\cdot)= p_{d}(\cdot ;\widehat{\boldsymbol{\theta}}_n) \widetilde{w}_{n}(\cdot ;\widehat{\boldsymbol{\theta}}_n)$ in (\ref{SME}) as
\begin{equation}\label{SME_KL}
\widehat{f}_{n}(\mathbf{x})\doteq 
\frac{1}{n} \sum_{i=1}^{n}\frac{p_0(\mathbf{x})}{p_0(\mathbf{X}_{i})}\mathbf{K}_{\mathbf{x},\mathbf{H}}(\mathbf{X}_i),\;\;\;\mathbf{x}\in\mathbb{T}_d^+.
\end{equation}
Thus, the following result provides approximate bias and variance. We omit its proof since it is analogous to the one of Proposition \ref{PropBiasVarf(x,0)}.

\begin{Pro}\label{PropBiasVarf(x,theta)}
	Let $p_0(\cdot):=p_{d}(\cdot ;\boldsymbol{\theta}_0)$  be the best $d$-variate approximant of the unknown pdmf $f(\cdot)= p_{d}(\cdot ;\boldsymbol{\theta}) w(\cdot;\boldsymbol{\theta})$ as (\ref{semiP}) under the Kullback–Leibler criterion, and let $w(\cdot) := f(\cdot)/p_0(\cdot)$ be the corresponding $d$-variate weight function. As $n\to\infty$ and under the assumption ({\bf a1}) on $f$, then the estimator $\widehat{f}_n (\cdot)=p_{d}(\cdot;\widehat{\boldsymbol{\theta}}_n)\widetilde{w}_n(\cdot;\widehat{\boldsymbol{\theta}}_n)$ in (\ref{SME}) of $f$ and refomulated in (\ref{SME_KL}) satisfies
	\begin{eqnarray*}\label{Biaisf(x,theta)} 
		\mathrm{Bias}\{\widehat{f}_n(\mathbf{x})\}&=&
		p_0(\mathbf{x})\left[w(\mathbf{x})-f(\mathbf{x})\{p_0(\mathbf{x})\}^{-1}+\left\langle\nabla w(\mathbf{x}),\mathbf{A}\left(\mathbf{x}, \mathbf{H}_n\right)\right\rangle\right] \\
		&&+
		\frac{1}{2}\; p_0(\mathbf{x})\left(\operatorname{tr} \left\{{\cal H} w\left(\mathbf{x}\right)\left[\mathbf{B}(\mathbf{x},\mathbf{H}_n)+\mathbf{A}\left(\mathbf{x},\mathbf{H}_n\right)^\mathsf{T}\mathbf{A}\left(\mathbf{x},\mathbf{H}_n\right)\right]\right\}\right)\\
		&&+\left(1+o\left\{\operatorname{tr}\left[\mathbf{B}(\mathbf{x},\mathbf{H}_n)\right]\right\}+n^{-2}\right),
	\end{eqnarray*}  
	for any $\mathbf{x}\in\mathbb{T}_d^+$.  Furthermore, if ({\bf a2}) holds then we have $\mathrm{var}\{\widehat{f}_n(\mathbf{x})\}
	=\mathrm{var}\{\widetilde{f}_n(\mathbf{x})\}$ of (\ref{Variance}).
\end{Pro}

Once again, the bias is different from that of (\ref{Biais}). Thus, the proposed semiparametric estimator $\widehat{f}_n$ in (\ref{SME}) of $f$  can be shown to be better or not than the traditional nonparametric one $\widetilde{f}_n$ in (\ref{NPE}). The following subsection provides a practical solution.

\subsection{Model Diagnostics}

The estimated weight function $\widetilde{w}_n(\mathbf{x},\widehat{\boldsymbol{\theta}}_n)$ given in (\ref{w_tilde}) provides useful information for model diagnostics. The $d$-variate  weight function $w(\cdot)$ is equal one if the $d$-variate parametric start model $p_d(\cdot;\boldsymbol{\theta}_0)$ is indeed the true pdmf. \citet{HjortG95} proposed
to check this adequacy by examining a plot of the weight function for various potential models with pointwise confidence bands to see wether or not $w(\mathbf{x})=1$ is reasonable. See also \cite{KSKB09,KZSK17} for univariate count setups.

In fact, without technical details here we use the model diagnostics for verifying the adequacy of the model by examining a plot of $\mathbf{x}\mapsto \widetilde{w}_n(\mathbf{x};\widehat{\boldsymbol{\theta}}_n)$ or 
\begin{equation}\label{Wn}
\widetilde{W}_n(\mathbf{x}):= \log\widetilde{w}_n(\mathbf{x};\widehat{\boldsymbol{\theta}}_n)=\log[\widehat{f}_n(\mathbf{x})/p_{d}(\mathbf{x};\widehat{\boldsymbol{\theta}}_n)]
\end{equation}
for all $\mathbf{x}=\mathbf{X}_i$, $i=1,\ldots,n$,  with a pointwise confidence band of $\pm1.96$ for large $n$;
that is to see how far away it is from zero. More precisely, for instance, $\widetilde{W}_n$ is $<5\%$ for pure nonparametric, it belongs to $[5\%,95\%]$ for semiparametric, and it is $>95\%$ for full parametric models. It is noteworthy that the retention of pure nonparametric means the inconvenience of parametric part considered in this approach; hence, the orthant dataset is left free. 

\section{Semicontinuous Examples of Application with Discussions}\label{5.Applic}

For a practical implementation of our approach, we propose to use the popular multiple gamma kernels as (\ref{KerProd}) by selecting the adaptive Bayesian procedure of \cite{SK20} to smooth $\widetilde{w}_n(\mathbf{x};\widehat{\boldsymbol{\theta}}_n)$. Hence, we shall gradually consider $d$-variate semicontinuous cases with $d=1,2,3$ for real datasets. All computations and graphics have been done with the \textsf{R} software \cite{R20}.  

\subsection{Adaptive Bayesian Bandwidth Selection for Multiple Gamma Kernels}

From Table \ref{Table_1sAK}, the function $G_{x,h}(\cdot)$ is the gamma kernel \cite{Chen00} given on the support $\mathbb{S}_{x,h}= [0,\infty)=\mathbb{T}_1^+$ with  $x\geq 0$ and $h>0$: 
\begin{equation*}\label{gam2}
G_{x,h}(u)=\dfrac{u^{x/h}}{\Gamma\left(1+x/h\right)h^{1+x/h}}\exp{\left(-\dfrac{u}{h}\right)}	\mathds{1}_{[ 0,\infty)}(u), 
\end{equation*} 
where $\mathds{1}_E$ denotes the indicator function of any given event $E$. This gamma kernel $G_{x,h}(\cdot)$ appears to be the pdf of the gamma distribution, denoted by $\mathcal{G}(1 + x/h,h)$ with shape parameter $1 + x/h$ and scale parameter $h$. The multiple gamma kernel from (\ref{KerProd}) is written as $\mathbf{K}_{\mathbf{x},\mathbf{H}}(\cdot)=\prod_{j=1}^{d} G_{x_j,h_{j}}(\cdot)$ with $\mathbf{H}=\mathrm{diag}_d\left(h_j\right)$. 

For applying (\ref{BayesAdap}) and (\ref{NPE_loo}) in the framework of semiparametric estimator $\widehat{f}_n$ in (\ref{SME}), we assume that each component $h_{i \ell}=h_{i \ell}(n)$, $\ell=1,\ldots,d$, of $\mathbf{H}_{i}$ has the univariate inverse gamma prior  $\mathcal{I}g(\alpha,\beta_{\ell})$ distribution with same shape parameters $\alpha > 0$ and, eventually, different scale parameters  $\beta_{\ell} > 0$ such that $\boldsymbol{\beta}=(\beta_1,\ldots,\beta_d)^\top$. We here recall that the pdf of $\mathcal{I}g(\alpha,\beta_\ell)$ with $\alpha,\beta_\ell>0$ is defined by 
\begin{equation}\label{prior}
Ig_{\alpha,\beta_\ell}(u)=\frac{\beta_\ell^{\alpha}}{\Gamma(\alpha)} u^{-\alpha-1} \exp(-\beta_\ell/u)\mathds{1}_{(0, \infty)}(u), \;\;\ell=1,\ldots,d.
\end{equation} 
The mean and the variance of the prior distribution \eqref{prior} for each component $h_{i\ell}$ of the vector $\mathbf{H}_i$ are given by $\beta_\ell/(\alpha -1)$ for $\alpha>1$, and $\beta_\ell^2/(\alpha-1)^2(\alpha-2)$ for $\alpha>2$, respectively.
Note that for a fixed $\beta_\ell>0$, $\ell=1,\ldots,d$, and if $\alpha \rightarrow \infty$, then the distribution of the bandwidth vector $\mathbf{H}_{i}$ is concentrated  around the null vector $\mathbf{0}$.  

From those considerations, the closed form of the posterior density and the Bayesian estimator of $\mathbf{H}_{i}$ are given in the following proposition which is proven in Appendix \ref{AppendixB}.

\begin{Pro}\label{PropBayesAdaG} 
	For fixed $i \in \{1,2,\ldots,n\}$, consider each observation $\mathbf{X}_{i}=(X_{i1},\ldots,X_{id})^{\top}$ with its corresponding $\mathbf{H}_{i}=\mathrm{diag}_d\left(h_{ij}\right)$ of univariate bandwidths and defining the subset  $\mathbb{I}_{i}=\left \{k \in \{1,\ldots,d\}~;X_{ik} = 0\right \}$ and its complementary set $\mathbb{I}^{c}_{i}=\left \{\ell \in \{1,\ldots,d\}~;X_{i\ell} \in (0, \infty)\right \}$. Using the inverse gamma prior $Ig_{\alpha,\beta_{\ell}}$ of (\ref{prior}) for  each component $h_{i\ell}$ of $\mathbf{H}_{i}$ in the multiple gamma estimator with $\alpha>1/2$ and $\boldsymbol{\beta}=(\beta_1,\ldots,\beta_d)^\top\in(0,\infty)^d$, then:
	
	(i) the posterior density is the following weighted sum of inverse gamma 
	\begin{eqnarray*}
	\pi(\mathbf{H}_{i}\mid\mathbf{X}_{i})&=&\frac{p_{d}(\mathbf{X}_i;\widehat{\boldsymbol{\theta}}_n)}{D_{i}(\alpha,\boldsymbol{\beta})}\sum_{j=1,j\neq i}^{n}\frac{1}{p_{d}(\mathbf{X}_{j};\widehat{\boldsymbol{\theta}}_n)}
	\left(\prod_{k \in \mathbb{I}_{i}}C_{jk}(\alpha,\beta_k)\,Ig_{\alpha+1,X_{jk}+\beta_{k}}(h_{ik})\right) \\
	&&\times \left(\prod_{\ell \in \mathbb{I}^{c}_{i}} A_{ij\ell}(\alpha,\beta_\ell)\,Ig_{\alpha+1/2,B_{ij\ell}(\beta_\ell)}(h_{i\ell})\right), 
	\end{eqnarray*}
	with $A_{ij\ell}(\alpha,\beta_\ell)= [\Gamma(\alpha +1/2)]/(\beta_{\ell}^{-\alpha}X_{i\ell}^{1/2}\sqrt{2\pi}[B_{ij\ell}(\beta_\ell)]^{\alpha +1/2})$, $B_{ij\ell}(\beta_\ell)= X_{i\ell}\log(X_{i\ell}/X_{j\ell})+X_{j\ell}-X_{i\ell}+\beta_{\ell}$, $C_{jk}(\alpha,\beta_k)= [ \Gamma ( \alpha +1) ]/[\beta_{k}^{-\alpha}(X_{jk}+\beta_{k})^{\alpha+1}]$, and\\ $D_{i}(\alpha,\boldsymbol{\beta})=p_{d}(\mathbf{X}_i;\widehat{\boldsymbol{\theta}}_n)\sum_{j=1,j\neq i}^{n}\left(\prod_{k \in \mathbb{I}_{i}}A_{ijk}(\alpha,\beta_k)\right)\left(\prod_{\ell \in \mathbb{I}^{c}_{i}} B_{ij\ell}(\beta_\ell)\right)/p_{d}(\mathbf{X}_j;\widehat{\boldsymbol{\theta}}_n)$;
	
	(ii) under the quadratic loss function, the Bayesian estimator $\widehat{\mathbf{\mathbf{H}}}_{i}=\mathrm{diag}_d\left(~\widehat{h}_{im}\right)$ of $\mathbf{H}_{i}$ in (\ref{SME}) is 
	\begin{eqnarray*}
	\widehat{h}_{im} &=& \frac{p_{d}(\mathbf{X}_i;\widehat{\boldsymbol{\theta}}_n)}{D_{i}(\alpha,\boldsymbol{\beta})}\sum_{j=1,j\neq i}^{n}\frac{1}{p_{d}(\mathbf{X}_{j};\widehat{\boldsymbol{\theta}}_n)}
	\left(\prod_{k \in \mathbb{I}_{i}}C_{jk}(\alpha,\beta_k)\right)
	\left(\prod_{\ell \in \mathbb{I}^{c}_{i}} A_{ij\ell}(\alpha,\beta_\ell)\right) \nonumber\\
	&&\times 	\left(\frac{X_{jm}+\beta_{m}}{\alpha}\mathds{1}_{\{0\}}(X_{im}) + \frac{B_{ijm}(\beta_m)}{\alpha-1/2}\mathds{1}_{(0,\infty)}(X_{im})\right),
	\end{eqnarray*}
	for $m=1,2,\ldots,d,$ with the previous notations of $A_{ij\ell}(\alpha,\beta_\ell)$, $B_{ijm}(\beta_m)$, $C_{jk}(\alpha,\beta_k)$ et $D_{i}(\alpha,\boldsymbol{\beta})$.
\end{Pro}

Following Som\'e and Kokonendji \cite{SK20} for nonparametric approach, we have to select the prior parameters $\alpha$ and $\boldsymbol{\beta}=(\beta_1,\ldots,\beta_d)^\top$ of the multiple inverse gamma of $\mathcal{I}g(\alpha,\beta_\ell)$ in (\ref{prior}) as follows: $\alpha = \alpha_n = n^{2/5}>2$ and $\beta_\ell>0$, $\ell=1,\ldots,d$, to obtain the convergence of the variable bandwidths to zero with a rate close to that of an optimal bandwidth. For practical use, we here consider each $\beta_\ell=1$.

\subsection{Semicontinuous Datasets}

The numerical illustrations shall be done through the following dataset of Table \ref{Table_WaterPump_data} which are recently used in \cite{SK20} for nonpaprametric approach and only in trivariate setup as semicontinuous data. It concerns three measurements (with $n=42$) of drinking water pumps installed in the Sahel. The first variable $X_1$ represents the \textit{failure times} (in months) and, also, it is  recently used by Tour\'e {\it et al.} \cite{TDGK20}. The second variable $X_2$ refers to the \textit{distance} (in kilometers) between each water pump and the repair center in the Sahel, while the third one $X_3$ stands for \textit{average volume} (in $m^3$) of water per day.
\begin{table}[H]
	\centering
	\caption{Drinking water pumps trivariate data measured in the Sahel with $n=42$.} 
	\begin{tabular}{lrrrrrrrrrrrrrr}
		\toprule
		$X_1:$&23&261&87&10&120&14&62&15&47&225&71&20&246&21\\
		$X_2:$&97&93&94&100&98&84&96&110&121&73&90&93&103&116\\
		$X_3:$&26& 52& 22& 39 &23 &26 &32& 17& 10 &39& 31& 42& 52& 26 \\ \midrule
		$X_1:$&19&42&20&5&12&120&17&11&3&14&71&11&5&14\\
		$X_2:$&114&82&96&94&77&91&117&103&99&113&79&109&84&118\\
		$X_3:$& 26& 36& 43& 36& 6& 27& 15 &36&  9& 52& 11& 20 &25& 37\\  \midrule
		$X_1:$&11&16&90&1&16&52&95&10&1&14&4&7&14&20\\
		$X_2:$&98&93&94&103&109&110&89&108&101&93&102&138&103&96\\
		$X_3:$& 25 &18 &43 &43 &24 &38 & 6 &40&21 &34 &15& 23 &68 &37\\
		\bottomrule
	\end{tabular}
	\label{Table_WaterPump_data}
\end{table}

Table \ref{Table_WaterPump_summ} displays all empirical univariate, bivariate and trivariate variation (\ref{GVI-def}) and dispersion (\ref{GDI-def}) indexes from Table \ref{Table_WaterPump_data}. Hence, each $X_j$, $(X_j,X_k)$ and $(X_1,X_2,X_3)$ is over-dispersed compared to the corresponding uncorrelated Poisson distribution. But, only $(X_1,X_3)$ (resp. $X_1$) can be considered as a bivariate equi-variation (resp. univariate over-variation) with respect to the corresponding uncorrelated exponential distribution; and, other $X_j$, $(X_j,X_k)$ and $(X_1,X_2,X_3)$ are under-varied. In fact, we just compute dispersion indexes for curiosity since all values in Table \ref{Table_WaterPump_data} are positive integers; and, we herenow omit the counting point of view in the remainder of the analysis.
\begin{table}[H]
	\centering
	\caption{Empirical univariate (in diagonal), bivariate (off diagonal) and trivariate (at the corner) variation and dispersion indexes.} 
	\begin{tabular}{rrrr|rrrr}
		\toprule
		$\widehat{\mathrm{GVI}}_3=0.0533$ & $X_1$ & $X_2$ & $X_3$ & $\widehat{\mathrm{GDI}}_3=15.1229$ & $X_1$ & $X_2$ & $X_3$ \\
		\midrule
		$X_1$ & $1.9425$ & $0.0557$ & $1.0549$ & $X_1$ & $89.5860$ & $14.3223$ & $70.7096$ \\	
		$X_2$ & $0.0557$ & $0.0167$ & $0.0157$ & $X_2$ & $14.3223$ & $1.6623$ & $2.0884$ \\	
		$X_3$ & $1.0549$ & $0.0157$ & $0.2122$ & $X_3$ & $70.7096$ & $2.0884$ & $6.3192$ \\		
		\bottomrule
	\end{tabular}
	\label{Table_WaterPump_summ}
\end{table}

Thus, we are gradually investing in semiparametric approaches for three univariates, three bivariates and only one trivariate from $(X_1,X_2,X_3)$ of Table \ref{Table_WaterPump_data}. 

\subsection{Univariate Examples}

For each univariate semicontinuous dataset $X_j$, $j=1,2,3$, we have already computed the GVI in Table \ref{Table_WaterPump_summ} which belongs in $(0.01,1.95)\ni 1$. This allows to consider our flexible  semiparametric estimation $\widehat{f}_{n,j}$ with an exponential $\mathscr{E}_1(\mu_j)$ as start in (\ref{SME}) and using adaptive Bayesian bandwidth in gamma kernel of Proposition \ref{PropBayesAdaG}. Hence, we deduce the corresponding diagnostic percent $\widetilde{W}_{n,j}$ from (\ref{Wn}) for deciding an appropriate approach. In addition, we first present the univariate nonparametric estimation $\widetilde{f}_{n,j}$ with adaptive Bayesian bandwidth in gamma kernel of \cite{Some20} and then propose another parametric estimation of $X_j$ by the standard gamma model with shape ($a_j$) and scale ($b_j$) parameters.

Table \ref{Table_WaterPump_estim1} transcribes parameter maximum likelihood estimates of exponential and gamma models with diagnostic percent from Table \ref{Table_WaterPump_data}. Figure \ref{fig_EstimatDiagnXi} exhibits histogram, $\widetilde{f}_{n,j}$, $\widehat{f}_{n,j}$, exponential, gamma and diagnostic $\widetilde{W}_{n,j}$ graphs for each univariate data $X_j$. One  can observe differences with the naked eye between $\widetilde{f}_{n,j}$ and $\widehat{f}_{n,j}$ although they are very near and with the same pace. The diagnostic $\widetilde{W}_{n,j}$ graphics lead to conclude to semiparametric approach for $X_2$ and to full parametric models for $X_3$ and slightly also for  $X_1$. Thus, we have suggested gamma model with two parameters for improving the starting exponential model; see, e.g., \cite[Table 2]{KTA20} for alternative parametric models.
\begin{table}[H]
	\centering
	\caption{Parameter estimates of models and diagnostic percents of univariate datasets.} 
	\begin{tabular}{crrrrrr}
		\toprule
		{\bf Estimate} & $\widehat{\mu}_j\;\;\;$  & $\widetilde{W}_{n,j}$ ($\%$)  & $\widehat{a}_j\;\;\;\;$ & $\widehat{b}_j\;\;\;\;$ \\
		\midrule
		$X_1$ & $0.0217$ & $95.2381$ & $0.7256$ & $63.5618$\\	
		$X_2$ & $0.0100$ & $76.1905$ & $56.9817$ & $1.7470$ \\	
		$X_3$ & $0.0336$ & $100.0000$ & $3.7512$ & $7.9403$\\	
		\bottomrule
	\end{tabular}
	\label{Table_WaterPump_estim1}
\end{table}

\begin{figure}
	\centering
	\mbox{
		\subfloat[]{	\resizebox*{7.5cm}{!}{\includegraphics{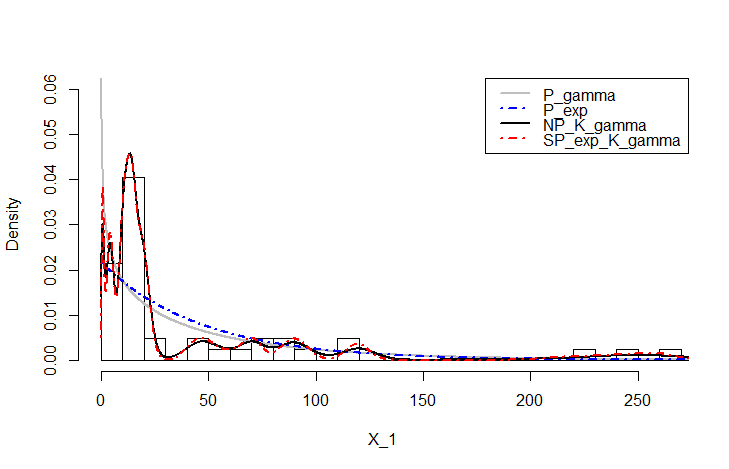}} }\hspace{3pt}
		\subfloat[]{\resizebox*{7.5cm}{!}{\includegraphics{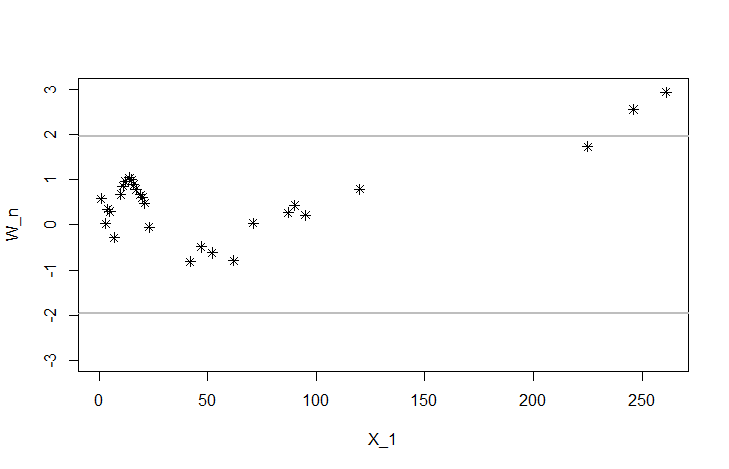}} }
	}
	\mbox{
		\subfloat[]{	\resizebox*{7.5cm}{!}{\includegraphics{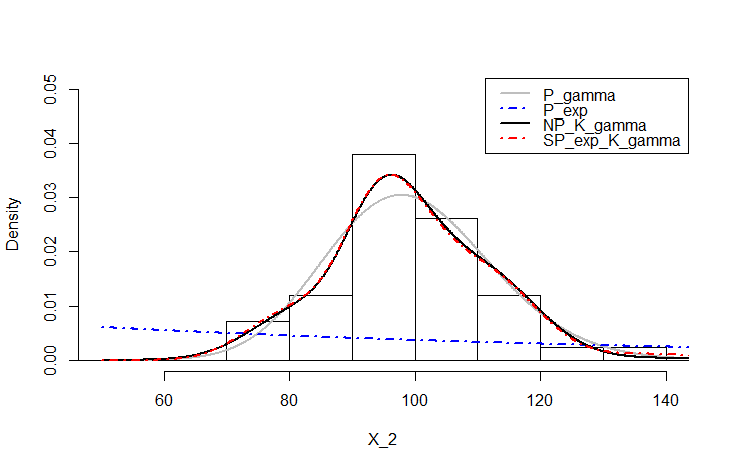}} }\hspace{3pt}
		\subfloat[ ]{\resizebox*{7.5cm}{!}{\includegraphics{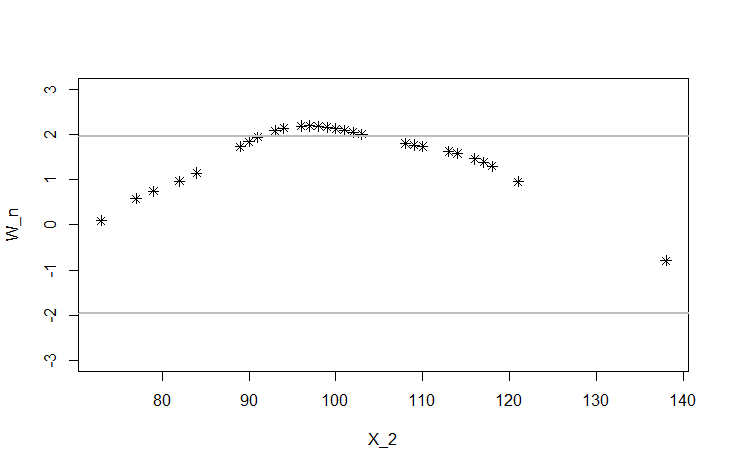}} } 
	}
	\mbox{
	\subfloat[]{	\resizebox*{7.5cm}{!}{\includegraphics{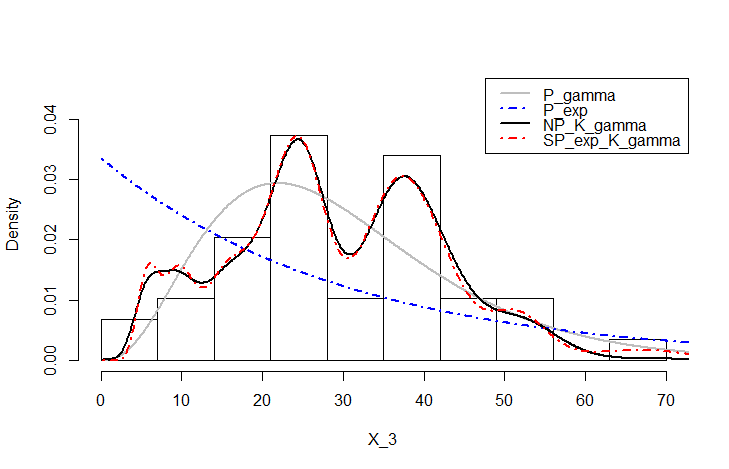}} }\hspace{3pt}
	\subfloat[ ]{\resizebox*{7.5cm}{!}{\includegraphics{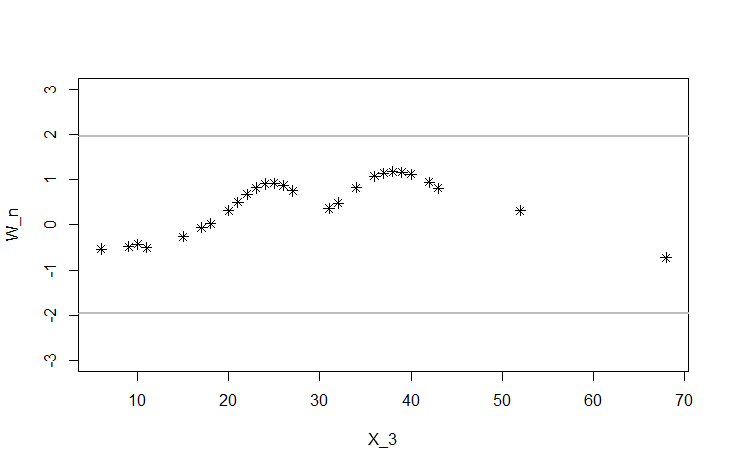}} } 
}
		\caption{Comparative graphs of estimates of $X_1$, $X_2$ and $X_3$ with their corresponding diagnostics.}\label{fig_EstimatDiagnXi}
\end{figure}

\subsection{Bivariate and Trivariate Examples}

For the sake of flexibility and efficiency, we here analyse our proposed semiparametric estimation $\widehat{f}_{n}$ with an uncorrolated exponential as start in (\ref{SME}) and using adaptive Bayesian bandwidth in gamma kernel of Proposition \ref{PropBayesAdaG}. This concerns all bivariate and trivariate datasets from Table \ref{Table_WaterPump_data} for which their GVI are in $(0.01,1.06)\ni 1$ from Table \ref{Table_WaterPump_summ}. All the computation times are alsmost instantaneous.
\begin{table}[H]
	\centering
	\caption{Correlations, MVI, parameter  estimates and diagnostic percents of bi- and trivariate cases.} 
	\begin{tabular}{lrrrrrr}
		\toprule
		\textbf{Dataset} & $(X_1,X_2)$ & $(X_1,X_3)$ & $(X_2,X_3)$ & $(X_1,X_2,X_3)$ \\
		\midrule
		$\widehat{\rho}(X_j,X_k)$ & $-0.3090$ & $0.2597$ & $0.0245$ & $\det\widehat{\boldsymbol{\rho}}=0.8325$\\ 
		$\widehat{\mathrm{MVI}}$ & $0.0720$ & $0.9857$ & $0.0155$ & $0.0634$\\
		$(\widehat{\mu}_j)$ & $(0.0217,0.0100)$ & $(0.0217,0.0336)$ & $(0.0100,0.0336)$ & $(0.0217,0.0100,0.0336)$\\
		$\widetilde{W}_{n}$ ($\%$) & $9.5238$ & $52.3809$ & $26.1005$ & $0.0000$\\	
		\bottomrule
	\end{tabular}
	\label{Table_WaterPump_estim23}
\end{table}

\begin{figure}[htp]
	\centering
	\mbox{
		\subfloat[]{	\resizebox*{7.5cm}{!}{\includegraphics{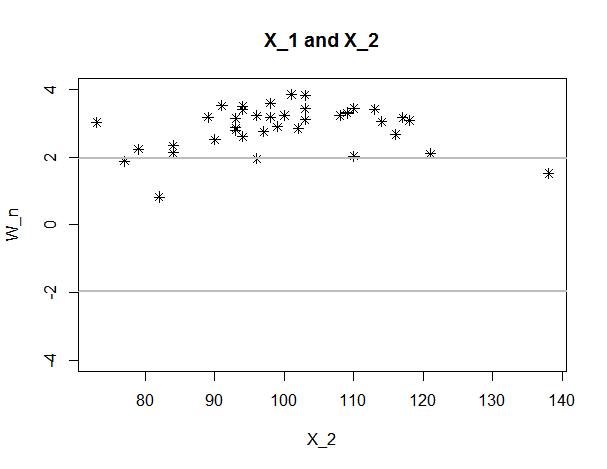}} }\hspace{3pt}
		\subfloat[]{\resizebox*{7.5cm}{!}{\includegraphics{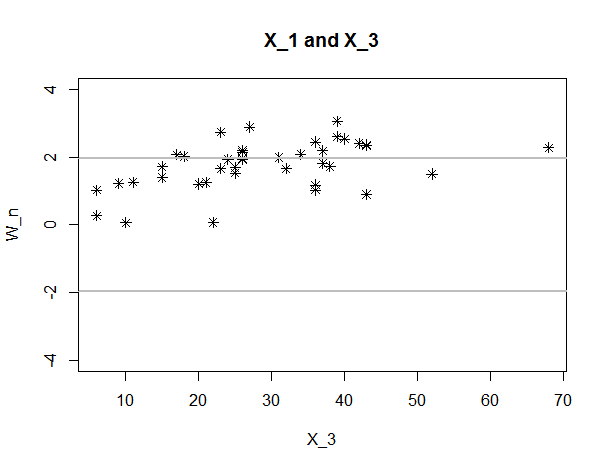}} }
	}
	\mbox{
		\subfloat[]{	\resizebox*{7.5cm}{!}{\includegraphics{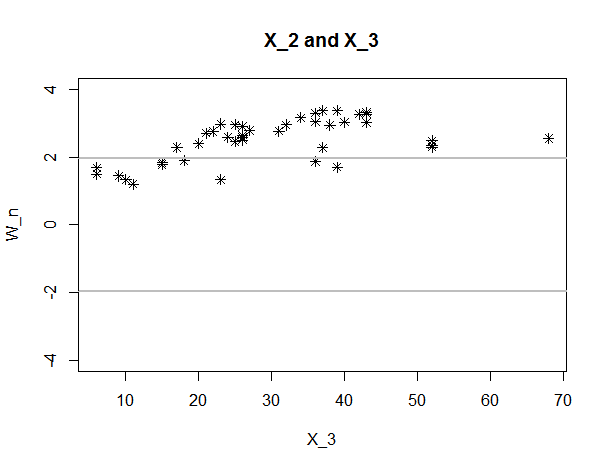}} }\hspace{3pt}
		\subfloat[ ]{\resizebox*{7.5cm}{!}{\includegraphics{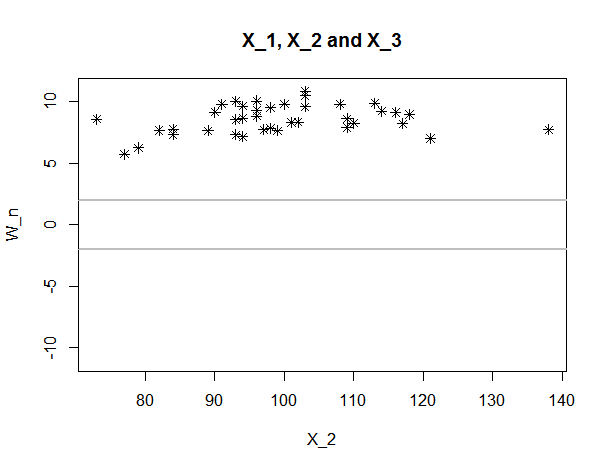}} } 
	}

	\caption{Univariate projections of diagnostic graphs for bivariate and trivariate models.}\label{fig_DiagnXijk}
\end{figure}

Table \ref{Table_WaterPump_estim23} reports the main numerical results of the corresponding correlations, MVI, parameter estimates and finally diagnostic percent $\widetilde{W}_{n}$ from (\ref{Wn}) that we intentionally omit to represent some graphics in three or four dimensions. However, Figure \ref{fig_DiagnXijk} displays some representative projections of $\widetilde{W}_{n}$. From Table \ref{Table_WaterPump_estim23}, the  cross empirical correlations are closed to 0 and all MVI are smaller than $1$ which allow to consider uncorrelated exponential start-models. The maximum likelihood method is also used for estimating the parameters $\mu_j$ for getting the same results as in Table \ref{Table_WaterPump_estim1}. Thus, the obtained numerical values of $\widetilde{W}_{n}$ indicate semiparametric approaches for all bivariate datasets and the purely nonparametric method for the trivariate one; see \citep{SK20} for more details on this nonparametric analysis. This progressive semiparametric analysis of the trivariate  dataset of Table \ref{Table_WaterPump_data} shows the necessity of a suitable choice of the parametric start-models, which may take into account the correlation structure. Hence, the retention of pure nonparametric means the inconvenience of parametric part used in the modelization. Note that we could consider the Marshall-Olkin exponential distributions with nonnegative correlations as start-models; but, they are singular. See Appendix \ref{AppendixA} for a brief review.

\section{Concluding Remarks}\label{6.Concl}

In this paper, we have presented a flexible semiparametric approach for multivariate nonnegative orthant distributions. We have first recalled multivariate variability indexes GVI, MVI, RVI, GDI, MDI and RDI from RWI as a prelude to the second-order discrimination for these parametric  distributions. We have then reviewed and provided new proposals to the nonparametric estimators through multivariate associated kernels; e.g., Proposition \ref{PropCst_n} and Corollary \ref{Coroll}. Both effective adaptive and local Bayesian selectors of bandwidth matrices are suggested for semicontinuous and counting data, respectively. 

All these previous ingredients were finally used to develop the semiparametric modelling for multivariate nonnegative orthant distributions. Numerical illustrations have been simply done for univariate and multivariate semicontinuous datasets with the uncorrolated exponential start-models after examining GVI and MVI. The adaptive Bayesian bandwidth selection (\ref{BayesAdap}) in multiple gamma kernel (Proposition \ref{PropBayesAdaG}) were here required for applications. Finally, the diagnostic models have played a very interesting role for helping to the appropriate approach, even if it means improving it later.

At the meantime, Kokonendji {\it et al.} \cite{Belaid20} proposed an in-depth practical analysis of multivariate count datasets starting by multivariate (un)correlated Poisson models after reviewing GDI and RDI. They have also established an equivalent of our Proposition \ref{PropBayesAdaG} for the local Bayesian bandwidth selection (\ref{BayesLoc}) by using the multiple binomial kernel from Example \ref{ExCountSt}. As one of the many perspectives, one could consider the categorial setup with local Bayesian version of the multivariate associated kernel of Aitchison and Aitken \cite{Aitchison76} from Example \ref{ExCateg} of the univariate case.

At this stage of analysis, all the main  foundations are now available for working in multivariate setup such as variability indexes, associated kernels, Bayesian selectors and model disgnostics. We just have to adapt them to each situation encountered. For instance, we have the semiparametric regression modelling; see, e.g., Abdous {\it et al.} \cite{AKSK12} devoted to count explanatory variables and \cite{SKZK16}. Also, an opportunity will be opened for hazard rate functions (e.g., \cite{Salha14}). The near future of other functional groups, such categorial and mixed, can now be considered with objectivity and feasibility.

\section{Appendix}\label{sec:Appendix}

\subsection{On a Broader $d$-Variate Parametric Models and the Marshall-Olkin Exponential}\label{AppendixA}

According to Cuenin {\it et al.} \cite{Cuenin16}, taking $p\in\{1,2\}$ in their multivariate Tweedie models of flexible dependence structure, another way to define the $d$-variate Poisson and exponential distributions is given by
$\mathscr{P}_d(\boldsymbol{\Lambda})$ and $\mathscr{E}_d(\boldsymbol{\Lambda})$,  respectively. The $d\times d$ symmetric variation matrix
$\boldsymbol{\Lambda}= (\lambda_{ij} )_{i,j \in \{1, \ldots, d\}}$
is such that $\lambda_{ij}=\lambda_{ji}\geq 0$, the mean of the corresponding marginal distribution is $\lambda_{ii}>0$, and the non-negative correlation terms satisfy
\begin{equation}\label{poisson_corr}
\rho_{ij} = \frac{\lambda_{ij}}{\sqrt{\lambda_{ii}\lambda_{jj}}}\in [0,\min\{R(i,j),R(j,i)\}),
\end{equation}
with
$R(i,j) = \sqrt{\lambda_{ii}/\lambda_{jj}} \, (1-\lambda_{ii}^{-1}\sum_{\ell\neq
	i,j}\lambda_{i\ell} )\in (0,1)$. Their  constructions are perfectly defined having $d(d+1)/2$ parameters as in $\mathscr{P}_d(\boldsymbol{\mu},\boldsymbol{\rho})$ and 
$\mathscr{E}_d(\boldsymbol{\mu},\boldsymbol{\rho})$. Moreover, we attain the exact bounds of the  correlation terms in (\ref{poisson_corr}). Cuenin {\it et al.} \cite{Cuenin16} have pointed out the construction and simulation of the negative correlation structure from the positive one of (\ref{poisson_corr}) by considering the inversion method.

The negativity of a correlation component is crucial for the phenomenon of under-variability in a bivariate/multivariate positive orthant  model. Figure~\ref{pict_corr} (right) plots a limit shape of any bivariate positive orthant distribution with very strong negative correlation (in red), which is not the diagonal line of the upper bound ($+1$) of positive correlation (in blue); see, e.g., \cite{Cuenin16} for details on both  bivariate orthant (i.e., continuous and count)  models. Conversely, Figure~\ref{pict_corr} (left) represents the classic lower ($-1$) and upper ($+1$) bounds of correlations on $\mathbb{R}^2$ or finite support.
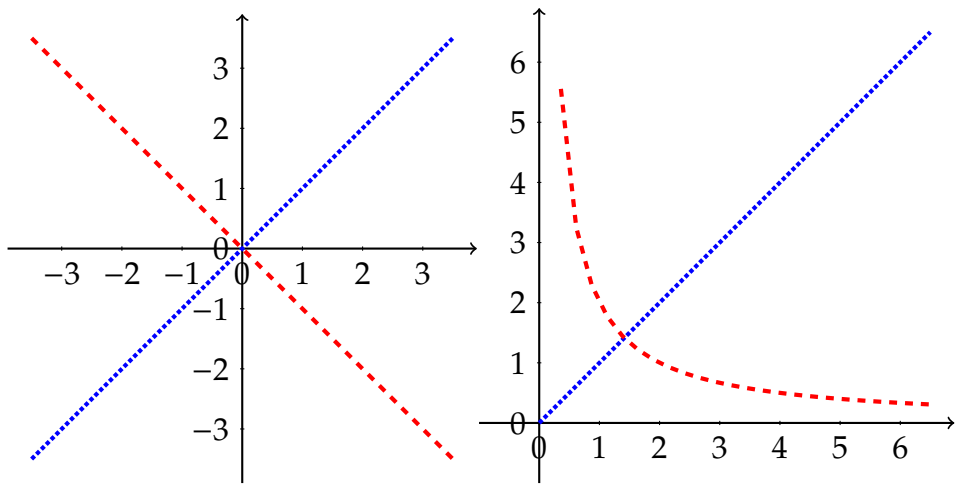
\begin{figure}[!htbp]
	\begin{center}
		\begin{tikzpicture}[scale=.8]
		\draw[thick,->](-3.9,0)--(3.9,0);
		\draw[thick,->](0,-3.9)--(0,3.9);
		\foreach \x in {-3,-2,-1,0,1,2,3}
		\draw(\x,1pt)--(\x,-1pt)node[below]{$\x$};
		\foreach \y in {-3,-2,-1,0,1,2,3}
		\draw(1pt,\y)--(-1pt,\y)node[left]{$\y$};
		\draw[blue,ultra thick,densely dotted][domain=-3.5:3.5] plot(\x,\x);
		\draw[red,ultra thick,dashed][domain=-3.5:3.5] plot(\x,-\x);
		\end{tikzpicture}
		\begin{tikzpicture}[scale=.8]
		\draw[thick,->](-1,0)--(6.9,0);
		\draw[thick,->](0,-1)--(0,6.9);
		\foreach \x in {0,1,2,3,4,5,6}
		\draw(\x,1pt)--(\x,-1pt)node[below]{$\x$};
		\foreach \y in {0,1,2,3,4,5,6}
		\draw(1pt,\y)--(-1pt,\y)node[left]{$\y$};
		\draw[blue,ultra thick,densely dotted][domain=0:6.5] plot(\x,\x);
		\draw[red,ultra thick,dashed][domain=0.36:6.5] plot(\x,2/\x);
		\end{tikzpicture}
		\caption{Support of bivariate distributions with maximum correlations (positive in blue and negative in red): model on $\mathbb{R}^2$ (left) and also finite support; model on $\mathbb{T}_2^+\subseteq [0,\infty)^2$ (right), without finite support.}
		\label{pict_corr}
	\end{center}
\end{figure}

The $d$-variate exponential $\boldsymbol{X}=(X_1,\ldots,X_d)^\top\sim\mathscr{E}_d(\boldsymbol{\mu}, \mu_0)$ of Marshall and Olkin \cite{Marshall67} is built as follows. Let $Y_1,\ldots, Y_d$ and $Z$ be univariate exponential random variables with parameters $\mu_1>0,\ldots, \mu_d>0$ and $\mu_0\geq 0$, respectively. Then, by setting $X_j:=Y_j+Z$ for $j=1,\ldots,d$, one has $\mathbb{E}X_j=1/(\mu_j+\mu_0)=\sqrt{\mathrm{var}X_j}$ and $\mathrm{cov}(X_j,X_\ell)=\mu_0/\{(\mu_j+\mu_0)(\mu_\ell+\mu_0)(\mu_j+\mu_\ell+\mu_0)\}$ for all $j\neq\ell$. Each correlation $\rho(X_j,X_\ell)=\mu_0/(\mu_j+\mu_\ell+\mu_0)$ lies in $[0,1]$ if and only if $\mu_0\geq 0$. From its survival (or reliability) function 
$$
S(\mathbf{x};\boldsymbol{\mu}, \mu_0)=\exp\left(-\mu_0\max (x_1,\ldots,x_d)-\sum_{j=1}^d\mu_j x_j\right),
$$
its pdf can be written as
\begin{equation*}\label{DensityMOexp}
p_d(\mathbf{x};\boldsymbol{\mu},\mu_0)=\left\{
\begin{array}{ll}
S(\mathbf{x};\boldsymbol{\mu}, \mu_0)(\mu_0+\mu_\ell)\prod\limits_{j=1,j\neq\ell}^d\mu_j & \mathrm{if}\;x_\ell:=\max (x_1,\ldots,x_d)\;\mathrm{and}\;x_\ell\neq x_j,\;j\neq\ell \\
S(\mathbf{x};\boldsymbol{\mu}, \mu_0)\mu_0\mu_{j_1}\cdots\mu_{j_k} & \mathrm{if}\;x_{j_1},\ldots,x_{j_k}<x_{\ell_{k+1}}=\cdots=x_{\ell_d} \\
S(\mathbf{x};\boldsymbol{\mu}, \mu_0)\mu_0 & \mathrm{if}\;x_1=\cdots=x_d>0.
\end{array}
\right.
\end{equation*}
It is not absolutely continuous with respect to the Lebesgue measure in $\mathbb{T}_d^+$ and has singularities corresponding to the cases where two or more of the $x_j$'s are equal. Karlis \cite{Karlis03} has proposed a maximum likelihood estimation of parameters via an EM algorithm. Finally, Kokonendji {\it et al.} \cite{KTS20} have calculated  
$$
\mathrm{GVI}(\boldsymbol{X})=1+\frac{\mu_0\sum_{j=1}^d(\mu_j+\mu_0)^{-1}\{\sum_{\ell\neq j}(\mu_j+\mu_\ell+\mu_0)^{-1}(\mu_\ell+\mu_0)^{-1}\}}{\{(\mu_1+\mu_0)^{-2}+\cdots +(\mu_d+\mu_0)^{-2}\}^2} \geq 1\;\;(\Leftrightarrow\mu_0\geq 0).
$$
and
$$
\mathrm{MVI}(\boldsymbol{X})=
\frac{\sum_{j=1}^d(\mu_j+\mu_0)^{-4}}{\sum_{j=1}^d(\mu_j+\mu_0)^{-4}+2\sum_{1\leq j<\ell\leq 1}(\mu_j+\mu_0)^{-2}(\mu_\ell+\mu_0)^{-2}}<1.
$$
Hence, the Marshall-Olkin exponential model $\boldsymbol{X}\sim\mathscr{E}_d(\boldsymbol{\mu}, \mu_0)$ is always under-varied with respect to the MVI and over- or equi-varied with respect to GVI. If $\mu_0=0$ then $\mathscr{E}_d(\boldsymbol{\mu}, \mu_0)$ is reduced to the uncorrolated $\mathscr{E}_d(\boldsymbol{\mu})$ with $
\mathrm{GVI}(\boldsymbol{Y})=1$. However, the assumption of non-negative correlations between components is sometimes unrealistic for some analyzes.

\subsection{Proofs of Proposition \ref{PropBiasVarf(x)}, Proposition \ref{PropBiasVarf(x,0)} and Proposition \ref{PropBayesAdaG}}\label{AppendixB}

\begin{proof}[Proof of Proposition \ref{PropBiasVarf(x)}]
	From Definition \ref{def_MDAK}, we get 
	(see also \cite{KS18} for more details)
	\begin{eqnarray}\label{Efntilde}
	\mathbb{E}\left[\widetilde f_n(\mathbf{x})\right]-f(\mathbf{x})&=&\mathbb{E}\left[\mathbf{K}_{\mathbf{x}, \mathbf{H}_n}(\mathbf{X}_{j})\right]-f(\mathbf{x})=\int_{\mathbb{S}_{\mathbf{x},\mathbf{H}_n}\cap \mathbb{T}_d^+}\mathbf{K}_{\mathbf{x},\mathbf{H}_n}(\mathbf{u})f(\mathbf{u})\boldsymbol{\nu}(d\mathbf{u})-f(\mathbf{x})\nonumber\\
	&=&\mathbb{E}\left[f\left({\cal Z}_{\mathbf{x},\mathbf{H}_n}\right)\right]-f(\mathbf{x}).
	\end{eqnarray}
	Next,  using (\ref{Efntilde}),  by a Taylor expansion of the function $f(\cdot)$  over   the points ${\cal Z}_{\mathbf{x},\mathbf{H}_n}$ and $\mathbb{E} \left[{\cal Z}_{\mathbf{x},\mathbf{H}_n}\right]$, we get 
	\begin{eqnarray}\label{Taylor2}
	f\left({\cal Z}_{\mathbf{x},\mathbf{H}_n}\right)&=&f\left(\mathbb{E}\left[{\cal Z}_{\mathbf{x},\mathbf{H}_n}\right]\right)+ \left\langle \nabla f\left(\mathbb{E}\left[{\cal Z}_{\mathbf{x},\mathbf{H}_n}\right]\right),\left({\cal Z}_{\mathbf{x},\mathbf{H}_n}-\mathbb{E}\left[{\cal Z}_{\mathbf{x},\mathbf{H}_n}\right]\right)\right\rangle\nonumber\\
	&&+\frac{1}{2} \left\langle {\cal H} f\left(\mathbb{E}\left[{\cal Z}_{\mathbf{x},\mathbf{H}_n}\right]\right)\left({\cal Z}_{\mathbf{x},\mathbf{H}_n}-\mathbb{E}\left[{\cal Z}_{\mathbf{x},\mathbf{H}_n}\right]\right),\left({\cal Z}_{\mathbf{x},\mathbf{H}_n}-\mathbb{E}\left[{\cal Z}_{\mathbf{x},\mathbf{H}_n}\right]\right)\right\rangle\nonumber\\
	&& +\left\Vert {\cal Z}_{\mathbf{x},\mathbf{H}_n}-\mathbb{E}\left[{\cal Z}_{\mathbf{x},\mathbf{H}_n}\right]\right\Vert^2o(1)\nonumber\\
	&=&f\left(\mathbb{E}\left[{\cal Z}_{\mathbf{x},\mathbf{H}_n}\right]\right)+ \left\langle \nabla f\left(\mathbb{E}\left[{\cal Z}_{\mathbf{x},\mathbf{H}_n}\right]\right),\left({\cal Z}_{\mathbf{x},\mathbf{H}_n}-\mathbb{E}\left[{\cal Z}_{\mathbf{x},\mathbf{H}_n}\right]\right)\right\rangle\nonumber\\
	&&+\frac{1}{2} \operatorname{tr} \left[{\cal H} f\left(\mathbb{E}\left[{\cal Z}_{\mathbf{x},\mathbf{H}_n}\right]\right)\left({\cal Z}_{\mathbf{x},\mathbf{H}_n}-\mathbb{E}\left[{\cal Z}_{\mathbf{x},\mathbf{H}_n}\right]\right)\left({\cal Z}_{\mathbf{x},\mathbf{H}_n}-\mathbb{E}\left[{\cal Z}_{\mathbf{x},\mathbf{H}_n}\right]\right)^\mathsf{T}\right]\nonumber\\
	&& + \operatorname{tr} \left[\left({\cal Z}_{\mathbf{x},\mathbf{H}_n}-\mathbb{E}\left[{\cal Z}_{\mathbf{x},\mathbf{H}_n}\right]\right)\left({\cal Z}_{\mathbf{x},\mathbf{H}_n}-\mathbb{E}\left[{\cal Z}_{\mathbf{x},\mathbf{H}_n}\right]\right)^\mathsf{T}\right]o(1),
	\end{eqnarray}
	where $o(1)$ is uniform in a neighborhood of $\mathbf{x}$.  
	Therefore, taking the expectation in both sides of \eqref{Taylor2} and then substituting the result in   \eqref{Efntilde},
	we get 
	\begin{eqnarray*} 
		\mathbb{E}\left[\widetilde f_n(\mathbf{x})\right]-f(\mathbf{x})&=&
		f\left(\mathbb{E}\left[{\cal Z}_{\mathbf{x},\mathbf{H}_n}\right]\right)-f(\mathbf{x})+\frac{1}{2} \operatorname{tr} \left[{\cal H} f\left(\mathbb{E}\left[{\cal Z}_{\mathbf{x},\mathbf{H}_n}\right]\right)\operatorname{var}\left({\cal Z}_{\mathbf{x},\mathbf{H}_n}\right)\right]\nonumber\\
		&&+o\left\{\operatorname{tr} \left[\operatorname{var}\left({\cal Z}_{\mathbf{x},\mathbf{H}_n}\right)\right]\right\}\nonumber\\
		&=&f\left(\mathbf{x}+\mathbf{A}\right)-f(\mathbf{x})+ \frac{1}{2} \operatorname{tr} \left[{\cal H} f\left(\mathbf{x}+\mathbf{A}\right)\mathbf{B}(\mathbf{x},\mathbf{H}_n)\right]+o\left\{\operatorname{tr}\left[\mathbf{B}(\mathbf{x},\mathbf{H}_n)\right]\right\},
	\end{eqnarray*}  
	where $o\left\{\operatorname{tr}\left[\mathbf{B}(\mathbf{x},\mathbf{H}_n)\right]\right\}$ is uniform  in a neighborhood of $\mathbf{x}$. The second Taylor expansion of the function $f(\cdot)$ over the points 
	$\mathbf{x}$ and $\mathbf{x}+\mathbf{A}\left(\mathbf{x},\mathbf{H}_n\right)$ allows to conclude the bias (\ref{Biais}).
	
	About the variance term, $f$ being  bounded, we have $\mathbb{E}\left[\mathbf{K}_{\mathbf{x}, \mathbf{H}_n}(\mathbf{X}_{j})\right]=O(1)$. It follows that:
	\begin{eqnarray*}
		\mathrm{var} \left[\widetilde f_n(\mathbf{x})\right]
		&=& \frac{1}{n}\mathrm{var}\left[\mathbf{K}_{\mathbf{x}, \mathbf{H}_n}(\mathbf{X}_{j})\right]\\
		&=&\frac{1}{n}\left[\int_{\mathbb{S}_{\mathbf{x},\mathbf{H}_n}\cap \mathbb{T}_d^+}\mathbf{K}^2_{\mathbf{x},\mathbf{H}_n}(\mathbf{u})f(\mathbf{u})\boldsymbol{\nu}(d\mathbf{u})+O(1)\right]\nonumber\\
		&=&\frac{1}{n} \int_{\mathbb{S}_{\mathbf{x},\mathbf{H}_n}\cap \mathbb{T}_d^+}\mathbf{K}^2_{\mathbf{x},\mathbf{H}_n}(\mathbf{u})\begin{pmatrix}f(\mathbf{x})+\left\langle\nabla f(\mathbf{x}),\mathbf{x}-\mathbf{u}\right\rangle\\+\frac{1}{2}(\mathbf{x}-\mathbf{u})^T{\cal H}f(\mathbf{x}) (\mathbf{x}-\mathbf{u})\\
			+o\left[\left(||\mathbf{x}-\mathbf{u}||^2\right)\right]
		\end{pmatrix}
		\boldsymbol{\nu}(d\mathbf{u})\\
		&=&\frac{1}{n}f(\mathbf{x})||\mathbf{K}_{\mathbf{x},\mathbf{H}_n}||_2^2+o\left[\frac{1}{n(\det\mathbf{H}_n)^r}\right].
	\end{eqnarray*}
\end{proof}

\begin{proof}[Proof of Proposition  \ref{PropBiasVarf(x,0)}]
	Since one has $\mathrm{Bias}[\widehat{f}_n(\mathbf{x})]=p_{d}(\mathbf{x};\boldsymbol{\theta}_0)\mathbb{E}[\widetilde{w}_n(\mathbf{x})]-f(\mathbf{x})$ and $\mathrm{var}[\widehat{f}_n(\mathbf{x})]=[p_{d}(\mathbf{x};\boldsymbol{\theta}_0)]^2 \mathrm{var} [\widetilde{w}_n(\mathbf{x})]$, it is enough to calculate $\mathbb{E}[\widetilde{w}_n(\mathbf{x})]$ and $\mathrm{var}[\widetilde{w}_n(\mathbf{x})]$ following Proposition \ref{PropBiasVarf(x)} applied to $\widetilde{w}_n(\mathbf{x})=n^{-1}\sum_{i=1}^n \mathbf{K}_{\mathbf{x},\mathbf{H}}(\mathbf{X}_i)/p_{d}(\mathbf{X}_i;\boldsymbol{\theta}_0)$ for all $\mathbf{x}\in\mathbb{T}_d^+$.
	
	Indeed, one successively has
	\begin{eqnarray*}
		\mathbb{E}[\widetilde{w}_n(\mathbf{x})] 
		&=&\mathbb{E}\left[\mathbf{K}_{\mathbf{x}, \mathbf{H}_n}(\mathbf{X}_{1})/p_{d}(\mathbf{X}_1;\boldsymbol{\theta}_0)\right]\\
		&=&\int_{\mathbb{S}_{\mathbf{x},\mathbf{H}_n}\cap \mathbb{T}_d^+}\mathbf{K}_{\mathbf{x},\mathbf{H}_n}(\mathbf{u})[p_{d}(\mathbf{u};\boldsymbol{\theta}_0)]^{-1}f(\mathbf{u})\boldsymbol{\nu}(d\mathbf{u})
		=\mathbb{E}\left[w\left({\cal Z}_{\mathbf{x},\mathbf{H}_n}\right)\right]\\
		&=&w(\mathbf{x})+\left\langle\nabla w(\mathbf{x}),\mathbf{A}\left(\mathbf{x}, \mathbf{H}_n\right)\right\rangle
		+
		\frac{1}{2}\left(\operatorname{tr} \left\{{\cal H} w\left(\mathbf{x}\right)\left[\mathbf{B}(\mathbf{x},\mathbf{H}_n)+\mathbf{A}\left(\mathbf{x},\mathbf{H}_n\right)^\mathsf{T}\mathbf{A}\left(\mathbf{x},\mathbf{H}_n\right)\right]\right\}\right)\\
		&&+o\left\{\operatorname{tr}\left[\mathbf{B}(\mathbf{x},\mathbf{H}_n)\right]\right\},
	\end{eqnarray*}
	which leads to the announced result of $\mathrm{Bias}[\widehat{f}_n(\mathbf{x})]$. As for $\mathrm{var}[\widetilde{w}_n(\mathbf{x})]$, one also write
	\begin{eqnarray*}
		\mathrm{var} \left[\widetilde w_n(\mathbf{x})\right]
		&=& \frac{1}{n}\mathrm{var}\left[\mathbf{K}_{\mathbf{x},\mathbf{H}_n}(\mathbf{X}_{1})/p_{d}(\mathbf{X}_1;\boldsymbol{\theta}_0)\right]\\
		&=&\frac{1}{n}\left[\int_{\mathbb{S}_{\mathbf{x},\mathbf{H}_n}\cap \mathbb{T}_d^+}\mathbf{K}^2_{\mathbf{x},\mathbf{H}_n}(\mathbf{u})[p_{d}(\mathbf{u};\boldsymbol{\theta}_0)]^{-2}f(\mathbf{u})\boldsymbol{\nu}(d\mathbf{u})+O(1)\right]\\
		&=&\frac{1}{n}f(\mathbf{x})[p_{d}(\mathbf{x};\boldsymbol{\theta}_0)]^{-2}||\mathbf{K}_{\mathbf{x},\mathbf{H}_n}||_2^2+o\left[\frac{1}{n(\det\mathbf{H}_n)^r}\right]
	\end{eqnarray*}
	and the desired result of $\mathrm{var}[\widehat{f}_n(\mathbf{x})]$ is therefore deduced.
\end{proof}

\begin{proof}[Proof of Proposition  \ref{PropBayesAdaG}]
	We have to adapt Theorem 2.1 of Som\'e and Kokonendji \cite{SK20} to this semiparametric estimator $\widehat{f}_n$ in (\ref{SME}). First, the leave-one-out associated kernel estimator (\ref{NPE_loo}) becomes
	\begin{equation*}\label{NPE_looS}
	\widehat{f}_{n,\mathbf{H}_i,-i}(\mathbf{X}_i):=
	\frac{p_{d}(\mathbf{X}_i;\widehat{\boldsymbol{\theta}}_n)}{n-1}\sum_{\ell=1,\ell\neq i}^{n}
	\frac{1}{p_{d}(\mathbf{X}_{\ell};\widehat{\boldsymbol{\theta}}_n)} \mathbf{K}_{\mathbf{X}_i,\mathbf{H}_i}(\mathbf{X}_\ell).
	\end{equation*}
	Then, the posterior distribution deduced from (\ref{BayesAdap}) is exppressed as
	\begin{equation*}\label{BayesAdapS}
	\pi(\mathbf{H}_i\mid\mathbf{X}_i):=\pi(\mathbf{H}_i)\widehat{f}_{n,\mathbf{H}_i,-i}(\mathbf{X}_i)\left[\int_{\mathcal{M}}\widehat{f}_{n,\mathbf{H}_i,-i}(\mathbf{X}_i)\pi(\mathbf{H}_i)d\mathbf{H}_i\right]^{-1}
	\end{equation*}
	and which leads to the result of Part (i) via \cite[Theorem 2.1 (i)]{SK20} for details. Consequently, we similarly deduce the adaptive Bayesian estimator $\widehat{\mathbf{\mathbf{H}}}_{i}=\mathrm{diag}_d\left(~\widehat{h}_{im}\right)$ of Part (ii).
\end{proof}

\section*{Acknowledgements}
We sincerely thank Mohamed Elmi Assoweh for some interesting discussions.


\section*{Abbreviations}
The following abbreviations are used in this manuscript:\\

\noindent 
\begin{tabular}{@{}ll}
	GDI & Generalized dispersion index\\
	GVI & Generalized variation index\\
	iid & Independent and identically distributed\\
	MDI & Marginal dispersion index\\
	MVI & Marginal variation index\\
	pdf & Probability density function\\
	pdmf & Probability density or mass function\\
	pmf & Probability mass function\\
	RDI & Relative dispersion index\\
	RVI & Relative variation index\\
	RWI & Relative variability index
\end{tabular}

\section*{References}



\newpage
\tableofcontents
\end{document}